\documentclass[11pt]{article}
\usepackage{fullpage} 
\usepackage{xspace}
\usepackage{comment}

\usepackage{amsmath}
\usepackage{amsthm}
\usepackage{amsfonts}
\usepackage{algorithm}
\usepackage{algpseudocode}

\usepackage{thmtools}
\usepackage{thm-restate}

\newtheorem{theorem}{Theorem}[section]
\newtheorem{lemma}[theorem]{Lemma}
\newtheorem{corollary}[theorem]{Corollary}

\theoremstyle{definition}
\newtheorem{definition}[theorem]{Definition}

\def\defeq{\stackrel{\mathrm{def}}{=}}

\usepackage{fancybox}

\newenvironment{fminipage}%
  {\begin{Sbox}\begin{minipage}}%
  {\end{minipage}\end{Sbox}\fbox{\TheSbox}}

\newenvironment{algbox}[0]{\vskip 0.2in
\noindent 
\begin{fminipage}{6.3in}
}{
\end{fminipage}
\vskip 0.2in
}

\usepackage{color}

\newcommand{\xmin}[0]{x_{\mathrm{min}}}

\DeclareMathOperator*{\argmin}{arg\,min}
\newcommand{\OV}{\textsc{OrthogonalVectors}\xspace}
\newcommand{\MinDeg}{\textsc{MinDegreeOrdering}\xspace}

\def\prob#1#2{\mbox{\textnormal{Pr}}_{#1}\left[ #2 \right]}
\def\expec#1#2{{\mathbb{E}}_{#1}\left[ #2 \right]}

\renewcommand\aa{\boldsymbol{\mathit{a}}}

\newcommand{\variable}[1]{\mathit{#1}}

\newcommand{\Exp}{\normalfont\text{Exp}}
\newcommand{\E}{\mathbb{E}}

\begin{document}

\title{
On Computing Min-Degree Elimination Orderings
}

\author{
Matthew Fahrbach\thanks{Supported in part by a National Science
Foundation Graduate Research Fellowship under grant DGE-1650044.}
\\
Georgia Tech\\
\texttt{matthew.fahrbach@gatech.edu}
\and
Gary L. Miller\thanks{This material is based on work supported by the
National Science Foundation under Grant No. 1637523.}\\
CMU\\
\texttt{glmiller@cs.cmu.edu}
\and
Richard Peng\thanks{This material is based on work supported by the
National Science Foundation under Grant No. 1637566.}\\
Georgia Tech\\
\texttt{rpeng@cc.gatech.edu}\and
Saurabh Sawlani\footnotemark[3]\\
Georgia Tech\\
\texttt{~~~~~sawlani@gatech.edu~~~~}
\and
Junxing Wang\\
CMU\\
\texttt{junxingw@cs.cmu.edu}
\and
Shen Chen Xu\footnotemark[2]\\
Facebook\thanks{Part of this work was done while at CMU.}\\
  \texttt{shenchex@cs.cmu.edu}
}
\maketitle

\begin{abstract}
We study faster algorithms for producing the minimum degree ordering used to
speed up Gaussian elimination.
This ordering is based on viewing the non-zero elements of a symmetric
positive definite matrix as edges of an undirected graph,
and aims at reducing the additional non-zeros (fill) in the matrix
by repeatedly removing the vertex of minimum degree.
It is one of the most widely used primitives for pre-processing
sparse matrices in scientific computing. 

Our result is in part motivated by the observation that sub-quadratic time
algorithms for finding min-degree orderings are unlikely,
assuming the strong exponential time hypothesis (SETH).
This provides justification for the lack of provably efficient
algorithms for generating such orderings, and leads us to study speedups
via degree-restricted algorithms as well as approximations.
Our two main results are:
(1) an algorithm that produces a min-degree ordering whose maximum
degree is bounded by $\Delta$ in $O(m \Delta \log^3{n})$ time, and
(2) an algorithm that finds an $(1 + \epsilon)$-approximate marginal
min-degree ordering in $O(m \log^{5}n \epsilon^{-2})$ time.

Both of our algorithms rely on a host of randomization tools related to the
$\ell_0$-estimator by [Cohen `97].
A key technical issue for the final nearly-linear time algorithm are the
dependencies of the vertex removed on
the randomness in the data structures.
To address this, we provide a method for generating a pseudo-deterministic
access sequence, which then allows the incorporation of data structures
that only work under the oblivious adversary model.
\end{abstract}

\pagenumbering{gobble}

\vfill

\pagebreak

\pagenumbering{arabic}

\section{Introduction}
\label{sec:Introduction}

Many algorithms in numerical analysis and scientific computing benefit from
speedups using combinatorial graph
theory~\cite{NaumannS12:book,HendricksonP06}.  Such connections are due to the
correspondence between non-zero entries of matrices and edges of graphs.  The
minimum degree algorithm is a classic heuristic for minimizing the space and
time cost of Gaussian elimination, which solves a system of linear equations by
adding and subtracting rows to eliminate variables.  As its name suggests, it
repeatedly pivots on the variable involved in the fewest number of
equations~\cite{GeorgeL89}.\footnote{We will assume the system is symmetric positive
definite (SPD) and thus the diagonal will remain strictly positive, allowing for any
pivot order.}  There are many situations where this is suboptimal.
Nonetheless, it is still a widely used and effective heuristic in practice
\cite{Amestoy2004,DGLN04}.  It is integral to the direct methods for solving
linear systems exactly in LaPack~\cite{Lapack}, which is in turn called by the
``\textbackslash'' command for solving linear systems in MATLAB~\cite{Matlab17}.
It is also a critical part of the linear algebra suite in Julia~\cite{Julia12}.

While the best theoretical running times for solving such systems either
rely on fast matrix multiplication~\cite{LeGall14} or iterative
methods~\cite{SpielmanTengSolver:journal,KoutisMP12}, direct methods
and their speedups are preferred in many cases.
For such elimination-based methods,
performances better than the general $O(n^3)$ bound for
naive Gaussian elimination are known only when the non-zero graph
has additional separators~\cite{LiptonTarjan79,LiptonRT79,GilbertT87} or
hierarchical structure~\cite{PouransariCD17}.
Nonetheless, these methods are still preferable for a variety of reasons.
They only depend on the non-zero structure,
and have fewer numerical issues.
More importantly, direct methods also benefit more from the inherent
sparsity in many real-world input instances.
For an input matrix and a given elimination order of the variables,
the non-zero structure that arises over the course of the elimination
steps has a simple characterization graph
theoretically~\cite{Rose1973,RoseTL76,LiptonRT79,GilbertT87}.

This characterization of additional non-zero entries, known as fill,
is at the core of elimination trees, which
allow one to precisely allocate memory for the duration of the
algorithm in $n\alpha(n)$ time~\cite{GilbertNP94}.
The reliable performance of elimination-based methods has led to
the study of elimination-based methods for solving more structured
linear systems~\cite{KyngS16}.
However, recent hardness results seem to indicate that speedups
via additional numerical structure may be limited to families of
specific problems instead of all sparse matrices arising
in scientific computing and numerical analysis~\cite{KyngZ17:arxiv}.

Although computing an elimination ordering that minimizes the total
cost is NP-hard in general~\cite{Berman90,Yannakakis81},
the minimum degree heuristic is exceptionally useful in practice.
When the non-zeros of the matrix are viewed as edges of a graph,
eliminating a vertex is equivalent to creating a clique on its
neighborhood and then deleting this vertex.
With this view in mind, the traditional min-degree algorithm can be viewed as:
(1) find the vertex $v$ with minimum degree
(which we term the fill-degree to avoid confusion with
the original graph)
in $O(n)$ time;
(2) add a clique among all its neighbors in $O(n^2)$ time;
(3) remove it together with all its edges from the graph in $O(n)$ time.

This leads to a running time that is $O(n^3)$---as high as
the cost of Gaussian elimination itself.
Somewhat surprisingly, despite the wide use of the min-degree
heuristic in practice, there have been very few works on
provably faster algorithms for producing this ordering.
Instead, heuristics such as AMD
(approximate-minimum degree ordering)~\cite{AmestoyDD96}
aim to produce orderings similar to minimum-degree orderings in provably faster
times such as $O(nm)$ without degree pivot size bounds.

Our investigation in this paper revolves around the question of finding
provably more efficient algorithms for producing exact and approximate
min-degree orderings.
We combine sketching with implicit representations of the fill
structure to obtain provably $O(nm \log{n})$ time algorithms.
These algorithms utilize representations of intermediate non-zero
structures related to elimination trees in order to implicitly
examine the fill, which may be much larger.
We also uncover a direct but nonetheless surprising connection between
finding min-degree vertices and popular hardness assumptions.
In particular, we show that computing the vertex of minimum degree after
several specified pivot steps cannot be done faster than $O(n^{2})$ time,
assuming the widely-believed strong exponential time
hypothesis~\cite{Williams05}.

Nevertheless, we are able to extend various tools from sketching and sampling
to give several improved bounds for computing and approximating minimum degree
orderings.
We show that our use of sketching can be much more efficient
when the maximum degree is not too large.
This in turn enables us to use sampling to construct data structures that
accurately approximate the fill-degrees of vertices in graphs in $\text{polylog}(n)$
time, even under pivoting of additional vertices.
Leveraging such approximate data structures, we obtain an algorithm for
producing an approximate marginal minimum degree ordering, which at each step
pivots a vertex whose degree is close to minimum, in nearly-linear time.
Our main result is:
\begin{theorem}
\label{thm:main}
Given an $n \times n$ matrix $A$ with non-zero graph structure $G$
containing $m$ non-zeros, we can produce an $\epsilon$-approximate greedy
min-degree ordering in $O(m \log^{5}n \epsilon^{-2})$ time.
\end{theorem}

Our algorithms combine classical ideas in streaming algorithms
and data structures, such as $\ell_0$-samplers~\cite{Cohen97},
wedge sampling~\cite{KallaugherP17,EdenLRS17}, 
and exponential start-time clustering~\cite{MillerPX13,MillerPVX15}.
Until now these tools have not been rigorously studied in the context of
scientific computing due to their dependency on randomization.
However, we believe there are many other algorithms and heuristics in
scientific computing that can benefit from the use of these techniques.

Furthermore, our overall algorithm critically relies on dissociating the
randomnesses from the pivot steps, as the update is dependent on the randomness
in the data structures.
In Section~\ref{subsec:Correlation} we give an example of how such correlations
can ``amplify'' errors in the data structures.
To address this issue, we define a pseudo-deterministic sequence of
pivots based on a second degree-estimation scheme, which we discuss
in Section~\ref{subsec:OverviewDecorrelation}.

Our paper is organized as follows.
We will formalize the implicit representation of fill
and definitions of exact, capped, and approximate min-degree
orderings in Section~\ref{sec:Preliminaries}.
Then in Section~\ref{sec:Overview} we give an overview of our
results and discuss our main decorrelation technique
in Subsection~\ref{subsec:OverviewDecorrelation}.
Our main hardness results are in Section~\ref{sec:Hardness},
while the use of sketching and sampling to obtain exact and
approximate algorithms are in Sections~\ref{sec:Sketching}
and~\ref{sec:Decorrelation}, respectively.
Further details on the graph theoretic building blocks
are in Sections~\ref{sec:DegreeEstimation}
and~\ref{sec:DynamicGraphs}. They respectively cover
the estimation of fill-degree of a single vertex
and the maintenance of sketches as vertices are pivoted.

\section{Preliminaries}
\label{sec:Preliminaries}

We work in the pointer model, where function arguments are pointers to
objects instead of the objects themselves.
Therefore, we do not assume that passing an object of size $O(n)$ costs $O(n)$
time and space.
This is essentially the ``pass by reference'' construct in
high-level programming languages.

\subsection{Gaussian Elimination and Fill}
\label{subsec:PreliminariesGaussian}

Gaussian elimination is the process of repeatedly
eliminating variables from a system of linear equations,
while maintaining an equivalent system on the remaining variables.
Algebraically, this involves taking one equation involving
some target variable and subtracting (a scaled version of) this equation
from all others involving the target variable.
Since our systems are SPD, we can also apply these operations to the columns
and drop the variable, which gives the Schur complement.

A particularly interesting fact about Gaussian elimination is that the
numerical Schur complement is unique irrespective of the ordering of pivoting.
Under the now standard assumption that non-zero elements do not cancel each
other out~\cite{GeorgeL89}, this commutative property also holds for the
combinatorial non-zero structure.  Since the non-zero structure of a matrix
corresponds to a graph, we can define the combinatorial change to the non-zero
structure of the matrix as a graph theoretic operation.  We start with the
notation from Gilbert, Ng, and Peyton~\cite{GilbertNP94}.
For a symmetric matrix, they use
\[
G(A)
\]
to denote the undirected graph formed by its non-zero structure.

Gilbert, Ng, and Peyton~\cite{GilbertNP94} worked with a known
elimination ordering and treated the entire fill pattern statically.
Because we work with partially eliminated states, we will need to distinguish
between the eliminated and remaining vertices in $G$ by implicitly associating
vertices with two states:
\begin{itemize}
\item Eliminated vertices will be denoted using $x$ and $y$.
\item Remaining vertices will be denoted using $u$, $v$, and $w$.
\end{itemize}
Then we use the \emph{fill graph}
\[
  G^{+}
\]
to denote the graph on the remaining vertices, where we add an edge $\{u,v\}$
between any pair of remaining vertices $u$ and $v$
connected via a path of eliminated vertices.  We
can also iteratively form the fill graph $G^{+}$ from $G$ by repeatedly removing an
eliminated vertex $w$ and its incident edges, and then adding edges between
all of the neighbors of $w$ to form a clique.  This characterization of fill
means that we can readily compute the fill-degree of a single vertex in a partially
eliminated state without explicitly constructing the matrix.

\begin{lemma}
\label{lem:ComputeFill}
For any graph $G$ and vertex $v \in V$,
given an elimination ordering $S$ we can compute in $O(m)$ time
the value $\deg(v)$ in $G^{+}$ when $v$ is eliminated.
\end{lemma}

\begin{proof}
Color the vertices in the sequence before $v$ red, and color all
remaining vertices green.
Run a depth-first search from $v$ that terminates at green vertices
$u \ne v$.
Let $D$ be the set of green vertices at which the search terminated.
It follows from the definition of $G^{+}$ that $\deg(v) = |D|$.
\end{proof}

This kind of path finding among eliminated vertices adds
an additional layer of complexity to our structures.
To overcome this, we contract eliminated vertices into their connected
components, leading to the notion of the \emph{component graph}.
We use
\[
  G^{\circ}
\]
to denote such a graph where we contract all edges $\{x,y\}$
between eliminated vertices $x$ and $y$.
We will denote the vertices corresponding to such components by $c$.
Note that $G^{\circ}$ is a quasi-bipartite graph, because the contraction rule
implies there are no edges between the component vertices.
It is also useful to denote the neighborhood of different kinds of vertices in
a component graph:
\begin{itemize}
\item  $N_{remaining}(c)$ or $N_{remaining}(u)$:
For a component $c$ or a remaining vertex $u$ in the component graph $G^{\circ}$,
we use $N_{remaining}(\cdot)$ to denote the neighbors 
that are remaining vertices.

\item $N_{component}(u)$: For a remaining vertex $u$,
this is the set of component vertices adjacent to $u$.

\item $N_{fill}(u)$: For a remaining vertex $u$,
this denotes the neighbors of $u$ in $G^{+}$,
which is 
\[
\left(
  \bigcup_{c \in N_{component}\left(u\right)}
    N_{remaining}\left( c \right)
\right)
\cup
N_{remaining} \left( u \right)
\cup
\left\{ u \right\}.
\]
\end{itemize}
Note that the fill-degree of a remaining vertex $u$ (its degree in $G^{+}$)
is precisely $|N_{fill}(u)|$.
Additionally, we use the restricted degrees:
\begin{itemize}
\item $d_{remain}(c)$ or $d_{remain}(u)$
to denote the size of $N_{remaining}(c)$ or $N_{remaining}(u)$,
respectively.
\item $d_{component}(u)$ to denote the size of $N_{component}(u)$
for some remaining vertex $u$.
\end{itemize}

\subsection{Min-Degree Orderings: Greedy, Capped, and Approximate}
For an elimination ordering
\[
  u_1, u_2, \dots, u_n,
\]
we define $G_i$ as the graph with vertices $u_1,u_2, \dots, u_{i}$
marked as eliminated
and $u_{i + 1}, u_{i+2},\dots, u_{n}$ marked as remaining.
Furthermore, we say such a permutation is a minimum degree permutation if at
each step $i$, the vertex $u_i$ has the minimum fill-degree in the non-zero
structure graph $G_{i - 1}$.
Concretely,
\begin{align}
\deg^{G_{i - 1}^{+}} \left( u_{i} \right)
=
\min_{v \in V\left( G_{i - 1}^{+} \right) }
\left\{
  \deg^{G_{i - 1}^{+}} \left( v \right)
\right\}.
\label{eq:MinDegree}
\end{align}

Because the performance of our algorithm degrades over time as the minimum
degree increases, we define the notion of a $\Delta$-capped minimum degree
ordering, where degrees are truncated to $\Delta$ before making a comparison.
We first define $\Delta$-capped equality where $\Delta$ is an integer.
\begin{definition}
\label{def:DeltaCappedEq}
  We use the notation $p =_{\Delta} q$ to denote $\min\{p, \Delta\} = \min\{q, \Delta\}$.
\end{definition}
\noindent
Now we can modify the definition of minimum degree in Equation~\ref{eq:MinDegree}
to specify that the elimination sequence $u_1, u_2, \dots, u_{n}$ satisfies the
$\Delta$-capped minimum degree property at each time step:
\begin{align}
\deg^{G_{i - 1}^{+}} \left( u_{i} \right)
=_{\Delta}
\min_{v \in V\left( G_{i - 1}^{+} \right) }
\left\{
  \deg^{G_{i - 1}^{+}} \left( v \right)
\right\}.
\label{eq:CappedMinDegree}
\end{align}

Our algorithm for finding the minimum ($\Delta$-capped) degrees is randomized,
so we need to be careful to not introduce dependencies between different steps
when several remaining vertices are of minimum degree.
To bypass this problem, we require that the lexicographically least vertex
be eliminated at each step in the event of a tie.
This simple condition is critical for arguing that our randomized routines do
not introduce dependencies as the algorithm progresses.

Lastly, our notion of approximating the min-degree ordering is based on finding
the vertex whose fill-degree is approximately minimum in the current graph
$G^+$.  This decision process has no look-ahead, and therefore does not in any
way approximate the minimum possible total fill.
\begin{definition}
\label{def:ApproxMinDegree}
An ordering of vertices $u_1, u_2, \dots, u_n$ is a $(1 +
  \epsilon)$-approximate greedy min-degree ordering if for all steps $1 \le i \le n$
  we have
\begin{align}
\deg^{G_{i - 1}^{+}} \left( u_{i} \right)
\leq
\left( 1 + \epsilon \right)
\min_{v \in V\left( G_{i - 1}^{+} \right) }
\left\{
  \deg^{G_{i - 1}^{+}} \left( v \right)
\right\}.
\label{eq:ApproxMinDegree}
\end{align}
\end{definition}

\subsection{Randomized Tools}
\label{subsection:WhereMistakesAreMade}

All of our algorithms are randomized, and
their analyses involve tools such as the union bound,
concentration bounds, and explicit calculations and
approximations of expected values.
We say an event happens with high probability (w.h.p.) if
for any constant $c > 0$ there is a setting of constants
(hidden by big-$O$ notation) so that this event occurs 
with probability at least $1-1/n^c$.
We also make extensive applications of backward
analysis~\cite{Siedel93}, which calculates the probabilities
of events locally using the current state of the data structures.

Our final algorithm for producing $\epsilon$-approximate
marginal min-degree orderings relies heavily on properties
of the exponential distribution in order to decorrelate updates
to the data structures and the results that it produces.
Properties of the exponential random variable are formalized
in Section~\ref{sec:Decorrelation}, and
we discuss its role in our algorithm in the overview
in Section~\ref{subsec:OverviewDecorrelation}.

The analysis of our algorithms critically hinges on viewing
all randomness as being generated before-hand, based on the
(potential) index in which the procedure gets called.
This is opposed to having a single source of randomness that
we query sequentially as the procedures are invoked.
For procedures such as the fill-degree estimator in
Section~\ref{subsec:DegreeEstimation_Matrix}, this method
leads to a simplified analysis by viewing the output of
a randomized sub-routine as a fixed distribution.
Such a view of randomization is also a core idea in our decorrelation
routine, which defines a random distribution on $n$ elements,
but only queries $O(1)$ of them in expectation.
This view is helpful for arguing that the randomness we 
query is independent of the indices that we ignored.

\subsection{Related Works}
\label{subsec:Related}

\subsubsection*{Fill from Gaussian Elimination and Pivot Orderings}

The study of better pivoting orderings is one of the
foundational questions in combinatorial scientific computing.
Work by George~\cite{George73} led to the study of
nested dissection algorithms, which utilize separators
to give provably smaller fill bounds for planar~\cite{RoseTL76,LiptonRT79}
and separable graphs~\cite{GilbertT87,AlonY10}.
One side effect of such a study is the far better
(implicit) characterization of fill entries
discussed in Section~\ref{subsec:PreliminariesGaussian}.
This representation was used to compute the total amount
of fill of a specific elimination ordering~\cite{GilbertNP94}.
It is also used to construct elimination trees,
which are widely used in combinatorial scientific computing
to both pre-allocate memory and optimize cache behaviors~\cite{Liu90}.

\subsubsection*{Finding Low Fill-in Orderings}

The ability to compute total fill for a given ordering raises
the natural question of whether orderings with near-optimal
fills can be computed.
NP-hardness results for finding the minimum fill-in
ordering~\cite{Yannakakis81,Berman90} were followed by
works for approximating the minimum total fill~\cite{NatanzonSS00},
as well as algorithms~\cite{KaplanST99,FominV13}
and hardness results for parameterized variants~\cite{WuAPL14,Bliznets16,CaoS17}.

Partially due to the higher overhead of these methods,
the minimum degree method remains one of the most widely used
methods for producing orderings with small fill~\cite{GeorgeL89}.
Somewhat surprisingly, we were not able to find prior works that
compute the exact minimum degree ordering in times faster
than $O(n^3)$, or ones that utilize the implicit representation
of fill provided by elimination trees.\footnote{
We use speculative language here due to the vastness
of the literature on variants of minimum degree algorithms.
}
On the other hand, there are various approximate schemes for
producing min-degree like orderings.
These include multiple minimum degree (MMD)~\cite{Liu85}
and an approximate minimum degree algorithm (AMD),
the latter of which is used in MATLAB~\cite{AmestoyDD96}.
While both of these methods run extremely well in practice,
theoretically they have tight performances of $O(n^2m)$ for MMD
and $O(nm)$ for AMD~\cite{HeggernesEKP01}.
Furthermore, AMD can be viewed as a different version of the
min-degree heuristic, as it is not always guaranteed to produce a vertex of
approximate minimum degree.

\subsubsection*{Estimating and Sketching Sizes of Sets}

The core difficulty of our algorithms is in estimating the cardinality
of sets (neighborhoods of eliminated components or component vertices
in component graphs $G^{\circ}$) under union and deletion of elements.
Many cardinality estimation algorithms have been proposed in the streaming
algorithm literature using similar ideas~\cite{FlajoletM85,CormodeM05}.
These algorithms often trade off accuracy for space, where as we trade
space for accuracy and efficiency in updates and queries.

Also closely related is another size-estimation framework for reachability
problems by Cohen~\cite{Cohen97}.
This work utilized $\ell_{0}$-estimators, which propagate random sketch values
along neighborhoods to estimate the size of reachable sets.
Our sketching method in Section~\ref{sec:Sketching} propagates the exact same
set of values.
However, we need to maintain this propagation under vertex pivots,
which is akin to contracting edges in the component graph.
This leads to a layer of intricacies that we resolve using amortized analysis
in Section~\ref{sec:DynamicGraphs}.

\subsubsection*{Removing Dependencies in Randomized Algorithms}

Lastly, our use of size estimators is dynamic---the choice of pivots, which in turn affects the subsequent graph eliminate
states, is a result of the randomness used to generate the results of previous
steps.
The independence between the access sequence and randomness is
a common requirement in recent works on data structures
that maintain spanning trees and matchings~\cite{BaswanaGS15,KapronKM13,Solomon16}.
There this assumption is known as the \emph{oblivious adversarial model},
which states that the adversary can choose the graph and the sequence of
updates, but it cannot choose updates adaptively in response to the randomly
guided choices of the algorithm.

There have been recent works that re-inject randomness to preserve
``independence'' of randomized dimensionality-reduction
procedures~\cite{LeeS15}.
The amount of ``loss'' in randomness has been characterized
via mutual information in a recent work~\cite{KapralovNPWWY17}.
Their bounds require an additional factor of $k$ of randomness
in order to handle $k$ adversarially injected information,
which as stated is too much for handling $n$ pivots adversarially.
Our work also has some tenuous connections to
recent works that utilize matrix martingales to analyze repeated introductions
of randomness in graph algorithms~\cite{KyngS16,KyngPPS17}.
However, our work utilizes more algorithmic tools than the martingale-based ones.

\section{Overview}
\label{sec:Overview}

The starting point of our investigation uses sketching
to design an efficient data structure for maintaining
fill-degrees under pivot operations.
This corresponds to edge contractions in the component graph and
is based on the observation that $\ell_0$-estimators propagate well along
edges of graphs.  For any $n \times n$ matrix with $m$ non-zero entries, this
algorithm takes $O(nm)$ time.

In our attempts to improve the running time of an exact algorithm,
we came to the somewhat surprising realization that
it is hard to compute the minimum degree in
certain partially eliminated graphs in time $O(n^{2 - \theta})$, for
any $\theta > 0$,
assuming the strong exponential time hypothesis.
We extend this observation to give super-linear hardness for computing minimum
degree orderings.

This hardness result for exact minimum degree sequences then motivated
us to parameterize the performance of min-degree algorithms in a new way.
Inspired by the behavior of AMD, we parameterize the performance of our
algorithm in terms of intermediate degrees.
Letting the minimum degree of the $i$-th pivot be $\Delta_i$
and the number of edges at that time be $m_i$,
we improve the performance of our algorithm to $O(\max_{i} m_i \Delta_i)$.
For many important real-world graphs such as grids and cube meshes, this bound
is sub-quadratic.
We then proceed to give a nearly-linear time algorithm for computing an
$\epsilon$-approximate marginal min-degree ordering, where at each step the
eliminated vertex has fill degree close to the current minimum.

\subsection{Sketching the Fill Graph}
\label{subsec:OverviewSketching}

We first explain the connection between computing fill-degrees and estimating
the size of reachable sets.
Assume for simplicity that no edges exist between the remaining vertices.
Consider duplicating the remaining vertices so that
each remaining vertex $u$ splits into $u_1, u_2$, and
any edge $\{u, x\}$ in the component graph becomes two directed edges
$(u_1 \rightarrow x)$
and
$(x \rightarrow u_2)$.
Then the fill-degree of~$u$ is the number of remaining
vertices $v_2$ reachable from $u_1$.
Estimating the size of reachable sets is a well-studied
problem for which Cohen~\cite{Cohen97} gave a nearly-linear time algorithm using
$\ell_0$-estimators.
Adapting this framework to our setting for fill
graphs (without duplication of vertices) leads to the
following $\ell_0$-sketch structure.
\begin{definition}
\label{def:Sketch}
An $\ell_0$-sketch structure consists of:
\begin{enumerate}
\item Each remaining vertex $u$ generating a random number $x_u$.
\item Each remaining vertex $u$ then computing the minimum $x_v$
among its neighbors in $G^{+}$ (including itself),
which is equivalent to
\[
\min_{v \in N_{reachable}\left( u \right)}
  x_{v}.
\]
\end{enumerate}
\end{definition}

In Section~\ref{sec:DynamicGraphs} we demonstrate that a copy
of this structure can be maintained efficiently through any
sequence of pivots in nearly-linear time.
As the priorities $x_u$ are chosen independently and uniformly at random, we
effectively assign each vertex $u$ a random vertex from its reachable set
$N_{reachable}(u)$.
Therefore, if we maintain $O(n \log{n})$ independent copies of this
$\ell_0$-sketch data structure, by a coupon-collector argument each vertex has
a list of all its distinct neighbors.
Adding together the cost of these $O(n \log{n})$ copies
leads to an $O(mn\log^{2}n)$ time algorithm for computing
a minimum degree sequence, which to the best of our knowledge is the fastest
such algorithm.

\subsection{SETH-Hardness of Computing Min-Degree Elimination Orderings}
\label{subsec:OverviewHardness}

Our hardness results for computing the minimum
fill degree and the min-degree ordering are based on
the \emph{strong exponential time hypothesis} (SETH),
which states that for all $\theta > 0$
there exists a~$k$ such that solving $k$-SAT requires
$\Omega(2^{(1-\theta)n})$ time.
Many hardness results based on SETH, including ours, go through the \OV problem
and make use of the following result.

\begin{theorem}[\cite{Williams05}]
\label{thm:HardnessOV}
Assuming SETH,
for any $\theta > 0$, there does not exist an $O(n^{2 - \theta})$ time algorithm
that takes $n$ binary vectors with $\Theta(\log^{2}{n})$ bits
and decides if there is an orthogonal pair.
\end{theorem}

\noindent
We remark that \OV is often stated as deciding if there exists
a pair of orthogonal vectors from two different sets~\cite{Williams2015hardness},
but we can reduce the problem to a single set by appending $[1; 0]$ to all
vectors in the first set and $[0; 1]$ to all vectors in the second set.

Our hardness observation for computing the minimum degree of a vertex in the
fill graph of some partially eliminated state is a direct reduction to \OV.
We give a bipartite graph construction that demonstrates how \OV can be
interpreted as deciding if a union of cliques covers a clique on the remaining
vertices of a partially eliminated graph.

\begin{lemma}
\label{lem:HardnessSingleStep}
Assuming SETH, for any $\theta > 0$,
there does not exist an $O(m^{2 - \theta})$ time algorithm
that takes any partially eliminated graph $G$ and computes the
minimum fill degree in $G^{+}$.
\end{lemma}

\begin{proof}
Consider an $\OV$ instance with $n$ vectors
$
  \aa(1), \aa(2), \dots, \aa(n) \in \{0,1\}^d.
$
Construct a bipartite graph $G = (V_{vec}, V_{dim}, E)$
such that each vertex in $V_{vec}$ corresponds to a vector $\aa(i)$
and each vertex in $V_{dim}$ uniquely corresponds to a dimension
$1 \leq j \leq d$.
For the edges, we connect vertices $i \in V_{vec}$
with $j \in V_{dim}$ if and only if $\aa(i)_j = 1$.

Consider the graph state with all of $V_{dim}$ eliminated
and all of $V_{vec}$ remaining.
We claim that there exists a pair of orthogonal vectors
among $\aa(1), \aa(2),\dots,\aa(n)$ if and only if
there exists a remaining vertex $v \in V(G^{+})$
with $\deg(v) < n - 1$.
Let $u, v \in V_{vec}$ be any two different vertices,
and let $\aa(u)$ and $\aa(v)$ be their corresponding vectors.
The vertices $u$ and $v$ are adjacent in $G^{+}$ if and only if there
exists a dimension $1 \leq j \leq d$
such that $\aa(u)_j = \aa(v)_j = 1$.

Suppose there exists an $O(m^{2 - \theta})$ time algorithm
for finding the minimum degree in a partially eliminated graph
for some $\theta > 0$.
Then for $d = \Theta(\log^2 n)$, we can use this algorithm to compute the
vertex with minimum fill degree in the graph described above in time
\[
  O\left(m^{2-\theta}\right)
  = O\left(\left(n \log^2 n\right)^{2-\theta}\right)
  = O\left(n^{2 - \theta / 2}\right),
\]
which contradicts SETH by Theorem~\ref{thm:HardnessOV}.
\end{proof}

In Section~\ref{sec:Hardness},
we extend this observation to show that
an $O(m^{4/3 - \theta})$ algorithm for computing
the min-degree elimination ordering does not exist,
assuming SETH.
This is based on constructing a graph where the
bipartite graph in the proof of Lemma~\ref{lem:HardnessSingleStep} appears in
an intermediate step.
The main overhead is adding more vertices and edges to
force the vertices in $V_{dim}$ to be eliminated first.
To do this, we first split such vertices into $\Theta(n)$
stars of degree $O(\sqrt{n})$.
Then we fully connect~$V_{vec}$ to an additional clique of size
$\Theta(\sqrt{n})$
to ensure that the (split) vertices in $V_{dim}$
are the first to be pivoted.
There are $O(n^{3/2}d)$ edges in this construction,
which leads to the $m^{4/3 - \theta}$-hardness.
However, we believe this is suboptimal and that 
$m^{2 - \theta}$-hardness is more likely.

\subsection{$\Delta$-capped and Approximately Marginal
Min-Degree Ordering}

This lower bound assuming SETH suggests that it is unlikely
to obtain a nearly-linear, or even sub-quadratic,
time algorithms for the min-degree ordering of a graph.
As a result, we turn our attention towards approximations
and output-sensitive algorithms.

Our first observation is that the size of
$N_{reachable}(u)$ can be bounded by $\Delta$,
so $O(\Delta \log{n})$ copies of the sketches as discussed
in Section~\ref{subsec:OverviewSketching} suffice for ``coupon collecting''
all  $\Delta$ distinct values instead of $O(n \log{n})$ copies.
This leads to bounds that depend on the maximum intermediate
fill-degrees, which on large sparse graphs are often significantly
less than $\Theta(n)$.
We also show how to maintain $O(\log{n})$ copies of the data structure and use
the $(1/e)$-th order statistic to approximate the number of entries in the set.
This leads to procedures that maintain approximate minimum
degree vertices for fixed sequences of updates.
This type of estimation is the same as using
$\ell_0$-estimators to approximate the size of reachable sets~\cite{Cohen97}.

This procedure of repeatedly pivoting out the approximate minimum
degree vertices given by sketching yields a nearly-linear time algorithm for
producing an $\epsilon$-approximate greedy min-degree ordering.
Initially, however, we were unable to analyze it because the input sequence is
not oblivious to the randomness of the data structure.
In particular, the choice of pivots is dependent on the randomness of the
sketches.
Compared to the other recent works that analyze sequential
randomness in graph sparsification~\cite{KyngS16,KyngPPS17}, our accumulation
of dependencies differs in that it affects the order in which vertices are
removed, instead of just the approximations in matrices.

\subsection{Correlation Under Non-Oblivious Adversaries}
\label{subsec:Correlation}

The general issue of correlations (or dependencies) between the
randomness of a data structure and access patterns to it can be remarkably
problematic.
We consider a simple example where deciding future updates based on the output
of previous queries results in a continual amplification of errors.
This can be understood as adversarially correlating the update sequence with
results of the randomness.
Consider the data structure in Figure~\ref{fig:ProbelmaticDS}
for maintaining a sequence of sets
\[
S_1, S_2, \dots, S_m \subseteq \left\{1,2,\dots,n\right\}
\]
under insertion/deletions and returns the one with minimum size 
up to an additive error of $\epsilon n$.
\begin{figure}[H]
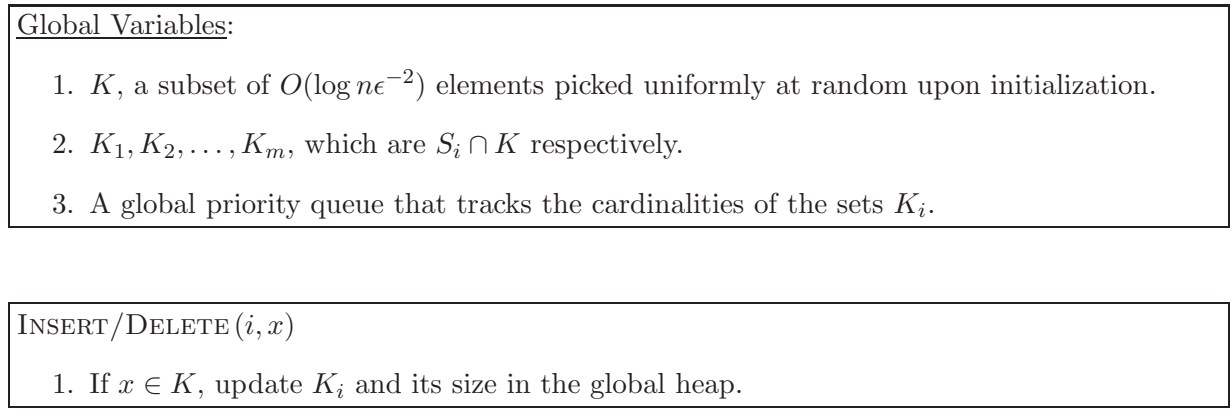

\begin{algbox}
  \underline{Global Variables}:
\begin{enumerate}
\item $K$, a subset of $O(\log{n} \epsilon^{-2} )$ elements
picked uniformly at random upon initialization.
\item $K_1, K_2, \dots, K_m$, which are $S_i \cap K$ respectively.
\item A global priority queue that tracks the cardinalities of the sets $K_i$.
\end{enumerate}
\end{algbox}
  \vspace{-0.05cm}
\begin{algbox}
$\textsc{Insert/Delete}\left(i, x\right)$ 
\begin{enumerate}
\item If $x \in K$, update $K_i$ and its size in the global heap.
\end{enumerate}
\end{algbox}
  \caption{Instance of a randomized data structure that can be adversarially correlated.}
\label{fig:ProbelmaticDS}
\end{figure}
For a non-adaptive sequence fixed ahead of time and a single set $S$,
Chernoff bounds give a result that is an $\epsilon n$
approximation with high probability.
Therefore, we can utilize this to build a data structure that maintains
a series of sets under insertion/deletion and returns a set of approximate
minimum cardinality (up to an additive $\epsilon n$).
Furthermore, to remove ambiguity, we assume this data structure
breaks ties lexicographically when the intersection of two sets with $K$ have
equal cardinality. 
With a similar invocation of Chernoff bounds, we can show that
this augmented data structure is correct under the oblivious adversary model.
As we maintain $k = O( \log{n} \epsilon^{-2})$ elements from each set $K_i$,
the total space usage of this data structure is $O(m \log{n} \epsilon^{-2})$.

On the other hand, an adaptive adversary can use the results of previous queries
to infer the set of secret keys $K$ in $O(n)$ queries.
Consider the following sequence of updates:
\begin{enumerate}
\item Start with two sets, $S_1$ and $S_2$, both initially equal to $\{1,2, \dots, n\}$.
\item For $x = 1,2, \dots, n$:
\begin{enumerate}
\item Delete $x$ from $S_2$.
\item If $S_2$ is the set of approximate minimum size (the one with
the smallest cardinality $|K_i|$),
insert $x$ back into $S_2$.
\end{enumerate}
\end{enumerate}
At the end of this sequence of updates, the only elements in $S_2$
are those in $K$, which is a substantially worse result than what we can
guarantee under the oblivious adversary model.

Our use of sketching to find a minimum degree vertex
clearly does not perform updates that are this adversarial, but
it does act on the minimum value generated by the randomized
routine so the final result can be reasonably inaccurate.
Moreover, any accounting of correlation (in the standard sense)
allows for the worst-case type of adaptive behavior described above.
In the next subsection, we describe an algorithmic approach to fix this issue.

\subsection{Decorrelating Sketches and Updates}
\label{subsec:OverviewDecorrelation}

Our correlation removal method is motivated by a third routine that
estimates the fill-degree of a remaining vertex in time
that is close to the degree of the vertex.
We then define an approximate, greedy min-degree sequence using
this routine.
At each step we choose the pivot vertex to be the minimizer of 
\[
\left( 1 - \frac{\epsilon \Exp\left( 1 \right) }{O\left( \log{n} \right)} \right)
\cdot \textsc{EstimateDegree}\left(u, \frac{\epsilon}{O\left( \log{n} \right)} \right),
\]
which is the $\epsilon$-decayed minimum over all the estimates
returned by the degree estimation routine.

We then utilize an $\ell_0$-estimation structure to maintain
approximate degrees throughout this update procedure.
By doing this, the randomness in the $\ell_0$-estimation data
structure is no longer correlated with the updates.
This sequence is defined with the randomness that is independent of the
$\ell_0$-estimators, and (after removing the probability of incorrectness)
may as well be considered deterministic.
On the other hand, evaluating such a sequence using only calls to
$\textsc{EstimateDegree}$ is expensive: it requires one call per
vertex, leading to a total of at least $\Omega(n^2)$.
Here we reincorporate the $\ell_0$-estimation data structure 
via the following observations about the initial perturbation
term involving the random variable $\Exp(1)$.
\begin{enumerate}
\item For a set of vertices whose degrees are
  within $1 \pm \epsilon/O(\log{n})$ of each other, 
it suffices to randomly select and consider $O(1)$ of them 
(by generating the highest order statistics for exponential random variables in
decreasing order).
\item By the memoryless property of the exponential distribution,
if we call $\textsc{EstimateDegree}$, with constant probability
it will be for the pivoted vertex.
Therefore, we can ``charge'' the cost of these evaluations to the overall
edge count and retain the nearly-linear time bounds.
\end{enumerate}

At a high level, we improve a data structure that only works
under the oblivious adversary model by providing it with a fixed
input using a second, more local, size-estimation routine.
Our generation of this ``fixed'' update sequence can still benefit
from the approximate bucketing created in the data structure.
The key idea is that any dependencies on the $\ell_0$-sketch structure stop
after these candidates are generated---their answers only depend on the
randomness of the separate size-estimation procedures.

This approach has close connections to pseudo-deterministic
algorithms~\cite{GatG11,Goldwasser12,GoldreichGR13},
which formalize randomized algorithms whose output sequences are fixed.
Such pseudo-deterministic update sequences seem particularly useful for
expanding the settings in which data structures designed for the oblivious
adversary model can be used.
We hope to formalize such connections in the near future.
However, the lack of a counterexample for directly using $\ell_0$-sketching
structures, or a proof of its correctness, suggests that some ideas are still
missing for the min-degree problem.


\section{SETH-Hardness of Computing Min-Degree Orderings}
\label{sec:Hardness}

We showed in Section~\ref{subsec:OverviewHardness} that computing the minimum
fill degree of a partially eliminated graph cannot be done in
$O(m^{2-\theta})$ time, for any $\theta > 0$, assuming the strong exponential
time hypothesis (SETH).
In this section, we augment this result to show that an exact
linear-time algorithm for computing min-degree elimination orderings is
unlikely.
In particular, our main hardness result is:

\begin{theorem}
\label{thm:HardnessOrdering}
Assuming SETH,
for any $\theta > 0$, there does not exist
an $O(m^{4/3 - \theta})$ time algorithm for producing a 
min-degree elimination ordering.
\end{theorem}

The main idea of our construction is to modify the bipartite graph
in Subsection~\ref{subsec:OverviewHardness} so that a minimum degree
ordering has the effect of necessarily eliminating
the $d$ vertices in $V_{dim}$ before any vector vertex in $V_{vec}$.
This allows us to use a minimum degree ordering on the graph to efficiently
solve an \OV instance.
The main bottleneck in our initial approach is that vertices in $V_{dim}$
can have degree as large as $n$,
so requiring that they are removed first is difficult.
We address this by breaking these vertices apart into $\Theta(n)$
vertices, each with degree $O(\sqrt{n})$,
using the following construction which we call a \emph{covering set system}.

\begin{lemma}
\label{lem:CoveringSetSystem}
Given any positive integer $n$,
we can construct in $O(n^{3/2})$ time
a covering set system of the integers $[n] = \{1,2, \dots, n\}$.
This system is collection of subsets $I_1, I_2, \dots, I_k \subseteq [n]$
such that:
\begin{enumerate}
  \item The number of subsets $k = O(n)$.
  \item The cardinality $|I_j| \le 10 \sqrt{n}$, for all $1 \leq j \leq k$.
  \item For each $(i_1, i_2) \in [n]^2$ there exists a subset $I_j$ such that
  $i_1, i_2 \in I_j$.
\end{enumerate}
\end{lemma}

Next we pad each of the vertices in $G_{vec}$ with $\Omega(\sqrt{n})$
edges to ensure that they are eliminated after the vertices
introduced by the covering set systems.
We outline this construction in Figure~\ref{fig:HardnessOrderingConstruction}.

\begin{figure}[H]

\begin{algbox}
\begin{enumerate}
\item Create one vertex per input vector $\aa(1),\aa(2),\dots,\aa(n)$,
and let these vertices be $V_{vec}$.
\item For each dimension $1 \leq j \leq d$:
  \begin{enumerate}
    \item Construct a covering set system for $[n]$.
  \item Create a vertex in $V_{dim}$ for each subset in this covering set
  system.
  \item For each vector $\aa(i)$ such that $\aa(i)_j = 1$,
    add an edge between its vertex in $V_{vec}$ and every vertex corresponding
    to a subset in this covering system that contains $i$.
  \end{enumerate}
\item Introduce $20 \sqrt{n}$ extra vertices called $V_{pad}$:
  \begin{enumerate}
  \item Connect all pairs of vertices in $V_{pad}$.
  \item Connect every vertex in $V_{pad}$ with every vertex in $V_{vec}$.
  \end{enumerate}
\end{enumerate}

\end{algbox}

\caption{Construction for reducing \OV to \MinDeg.}

\label{fig:HardnessOrderingConstruction}
\end{figure}

\begin{lemma}
\label{lem:size}
Let $G$ be the graph
produced by the construction in Figure~\ref{fig:HardnessOrderingConstruction}
for an instance of \OV with $n$ vectors of dimension $d$.
We have $|V| = O(nd)$ and $|E| = O(n^{3/2}d)$.
\end{lemma}
\begin{proof}
The number of vertices in $G$ is
\[
  \left|V\right|
  =
  20 \sqrt{n} + n + d \cdot O\left(n\right)
  = O\left(nd\right).
\]
Similarly, an upper bound on the number of edges in $G$ is
\[
 \left|E\right|
  = \binom{20 \sqrt{n} }{2} + 20 \sqrt{n} \cdot n
  + d \cdot 10\sqrt{n} \cdot O\left(n\right) 
  = 
  O\left(n^{3/2} d\right),
\]
where the terms on the left-hand side of the final equality correspond to edges
contained in $V_{pad}$,
the edges between $V_{pad}$ and $V_{vec}$,
and edges between $V_{vec}$ and $V_{dim}$, respectively.
\end{proof}

\begin{lemma}
\label{lem:ConstructionElimOrdering}
Consider a graph $G$ constructed from an \OV instance
as described in Figure~\ref{fig:HardnessOrderingConstruction}.
For any min-degree ordering of $G$, the first
vertices to be eliminated are those in $V_{dim}$.
The fill degree of the next eliminated vertex is $\min_{v \in V_{vec}} \deg(v)$.
\end{lemma}

\begin{proof}
Let the graph be $G = (V, E)$, where $V$ is partitioned into
\[
  V_{vec} \cup V_{dim} \cup V_{pad}
\]
as described in Figure~\ref{fig:HardnessOrderingConstruction}.
Initially, for every vertex $v_{pad} \in V_{pad}$ we have
\[
  \deg\left( v_{pad} \right)
  =
  \left(20 \sqrt{n} - 1\right) + n.
\]
For every vertex $v_{vec} \in V_{vec}$ we have
\[
  \deg\left(v_{vec}\right)
  =
  20 \sqrt{n} + \left| E\left( v_{vec}, V_{dim} \right) \right|
  \geq
  20 \sqrt{n},
\]
and for every vertex $v_{dim} \in V_{dim}$ we have
\[
  \deg\left(v_{dim} \right) \le 10 \sqrt{n}.
\]

Pivoting out a vertex in $V_{dim}$ does not increase the degree of any other
vertex in $V_{dim}$, because no two vertices in $V_{dim}$ are adjacent.
As these vertices are pivoted, we still maintain
\[
  \deg(v) \ge 20 \sqrt{n},
\]
for all $v \in V_{vec}$.
Therefore, the first vertices to be pivoted must be all $v \in V_{dim}$.
After all vertices in $V_{dim}$ have been pivoted,
the next vertex must have fill degree $\min_{v \in V_{vec}} \deg(v)$,
because either a vertex in $V_{vec}$ will be eliminated
or all remaining vertices have fill degree $20\sqrt{n} + n - 1$.
\end{proof}

\begin{proof}[Proof of Theorem~\ref{thm:HardnessOrdering}.]
Suppose for some $\theta > 0$ there exists an $O(m^{4/3 - \theta})$ time algorithm
for \MinDeg.
Construct the graph $G=(V,E)$ with covering sets as described in
Figure~\ref{fig:HardnessOrderingConstruction}.
For $d = \Theta(\log^2 n)$,
it follows from Lemma~\ref{lem:size} that
$|V|=O(n \log^2 n)$
and
$|E| = O(n^{3/2} \log^2 n)$.
Therefore by the assumption, we can obtain
a min-degree ordering of $G$ in time
\[
  O\left(m^{4/3 - \theta}\right)
  = O\left( \left(n^{3/2} \log^2 n\right)^{4/3 - \theta}\right)
  = O\left(n^{2 - \theta}\right).
\]

By Lemma~\ref{lem:ConstructionElimOrdering}, the state
of the elimination steps after the first $|V_{dim}|$ vertices
have been pivoted is essentially identical to the
partially eliminated state from Lemma~\ref{lem:HardnessSingleStep}.
Then by Lemma~\ref{lem:ComputeFill},
we can compute the degree of the next vertex
to be eliminated in $O(m) = O(n^{2 - \delta})$ time.
Checking whether the degree of that vertex is $20\sqrt{n} + n - 1$
allows us to solve \OV in
time $O(n^{2-\theta})$, which contradicts SETH.
\end{proof}

It remains to efficiently construct the covering set systems as defined in
Lemma~\ref{lem:CoveringSetSystem}, which we can interpret as a strategy for
covering all the edges of $K_n$ with $O(n)$ $K_{10\sqrt{n}}$ subgraphs.
We also note that our construction of covering set systems is related to
existence results for the covering problem with fixed-size
subgraphs~\cite{Chee2013covering,Caro1998covering}.

\begin{proof}[Proof of Lemma~\ref{lem:CoveringSetSystem}.]
Let $p = \textsc{NextPrime}(\sqrt{n})$.
Bertrand's postulate asserts that
$p < 4\sqrt{n}$,
so we can compute $p$ in $O(n)$ time.
Clearly we have $[n] \subseteq [p^2]$, so it suffices to find a covering for $[p^2]$.
Map the elements of $[p^2]$ to the coordinates of a $p \times p$ array in the
canonical way so that
\begin{align*}
  1 &\mapsto (0,0)\\
  2 &\mapsto (0, 1)\\
  &\hspace{0.22cm}\vdots\\
  p^2 &\mapsto (p-1,p-1).
\end{align*}
For all $(a,b) \in \{0,1,\dots,p-1\}^2$,
define
\[
D\left(a,b\right)
\defeq
\left\{
  \left(x,y\right) \in \left\{0,1,\dots,p-1\right\}^2
  : y \equiv ax + b \pmod{p}
\right\}
\]
to be the diagonal subsets of the array,
and define
\[
R\left(a\right)
\defeq
\left\{
  \left(x,y\right) \in \left\{0,1,\dots,p-1\right\}^2 :
  x \equiv a \pmod{p}
\right\}
\]
to be the row subsets of the array.
Let the collection of these subsets be
\[
  S
  =
  \left\{
    D\left(a,b\right) : a,b \in \left\{0,1,\dots,p-1\right\}
  \right\}
  \cup
  \left\{
     R\left(a\right) : a \in \left\{0,1,\dots,p-1\right\}
  \right\}.
\]

The construction clearly satisfies the first two conditions.
Consider any $(a,b) \in [p^2]^2$ and their coordinates
in the array $(x_1,y_1)$ and $(x_2,y_2)$.
If $x_1 = x_2$, then $(x_1,y_1), (x_2,y_2) \in R(x_1)$.
Otherwise, 
it follows that $(x_1,y_1)$ and $(x_2,y_2)$ are solutions to the line
\[
y
\equiv
  \frac{y_1 - y_2}{x_1 - x_2} \cdot
\left( x - x_1 \right) + y_1 \pmod{p},
\]
so the third condition is satisfied.
\end{proof}

\section{Sketching Based Algorithms for Computing Degrees}
\label{sec:Sketching}

Let us recall a few relevant definitions from Section~\ref{sec:Preliminaries} for convenience.
For a given vertex elimination sequence
\[
  u^{(1)}, u^{(2)}, \ldots, u^{(n)},
\]
let ${G^+}^{(t)}$ denote the fill graph obtained by pivoting
vertices $u^{(1)}, u^{(2)}, \ldots, u^{(t)}$.
Let $\delta^{(t)}$ denote the minimum degree of a vertex in ${G^+}^{(t)}$.
An $\ell_0$-sketch data structure consists of the following:
\begin{itemize}
  \item Each remaining vertex $u$ generates a random number $x_u$.
  \item Each remaining vertex $u$ computes the vertex with the
  minimum $x_v$ value among its neighbors in ${G^+}^{(t)}$ and itself
  (which we call the minimizer of $u$).
\end{itemize}

In this section we show that
if an $\ell_0$-sketch data structure can be maintained efficiently
for a dynamic graph, then we can use a set of copies of this data structure
to find the vertex with minimum fill degree 
at each step and pivot out this vertex.
Combining this with data structures for efficiently propagating
sketch values from Section~\ref{sec:DynamicGraphs} gives a faster
algorithm for computing minimum degree orderings on graphs.
We use this technique in three different cases.

First, we consider the case where the minimum degree at each step is bounded.
In this case,
we choose a fixed number of copies of the $\ell_0$-sketch data structure
and look at the minimizers over all the copies.
\begin{theorem}
	\label{thm:BoundedDegreeAlgo}
	There is an algorithm \textsc{DeltaCappedMinDegree} that,
	when given a graph with a lexicographically-first
	min-degree ordering such that the minimum degree is always bounded by $\Delta$,
	outputs the ordering with high probability in expected time $O(m \Delta \log^3 n)$
	and uses space $O(m \Delta \log n )$.
\end{theorem}
	
Next, we eliminate the condition on the minimum degrees
and allow the time and space bounds of the algorithm to be output sensitive.
In this case, we adaptively increase the number of copies of the $\ell_0$-data
structure.

\begin{theorem}
  \label{thm:OutputSensitiveAlgo}
  There is an algorithm \textsc{OutputSensitiveMinDegree} that,
  when given a graph with a lexicographically-first
  min-degree ordering $\delta^{(1)}, \delta^{(2)}, \ldots, \delta^{(n)}$,
  outputs this ordering with high probability in expected time $ O(m \cdot \max_{1 \leq t \leq n} \delta^{(t)} \cdot \log^3 n)$
  and uses space $O(m \cdot \max_{1 \leq t \leq n} \delta^{(t)} \cdot \log n)$.
\end{theorem}

Lastly, we modify the algorithm to compute the
approximate minimum degree at each step.
In this case, we use $O(\log n \epsilon^{-2})$ copies
of the data structure and use the reciprocal of 
the $(1-1/e)$-th percentile among the $x_v$ values of its minimizers
as an effective approximate of the vertex degree.

\begin{theorem}
  \label{thm:ApproxDegreeDS}
  There is a data structure $\textsc{ApproxDegreeDS}$
  that supports the following two operations:
  \begin{enumerate}
    \item $\textsc{ApproxDegreeDS\_Pivot}(u)$, which pivots a remaining vertex $u$.
    \item $\textsc{ApproxDegreeDS\_Report}()$, which provides
    balanced binary search tree (BST) containers
      \[
        V_1, V_2, \dots, V_{B}
      \]
    such that all vertices in the bucket $V_{i}$ have degrees in the range
    \[
    \left[
    \left( 1 + \epsilon \right)^{i - 2},
    \left( 1 + \epsilon \right)^{i + 2}
    \right].
    \]
  \end{enumerate}
  
  \noindent
  The memory usage of this data structure is $O(m \log{n} \epsilon^{-2})$.
  Moreover, if the pivots are picked independently from the
  randomness used in this data structure
  (i.e., we work under the oblivious adversary model) then:
  \begin{enumerate}
    \item The total cost of all the calls to
       $\textsc{ApproxDegreeDS\_Pivot}$
      is bounded by $O(m \log^3{n}
      \epsilon^{-2})$.
    \item
      The cost of each call to $\textsc{ApproxDegreeDS\_Report}$
    is bounded by $O(\log^2{n} \epsilon^{-1})$.
  \end{enumerate}
\end{theorem}

\subsection{Computing Exact Min-Degree}
We consider the case where the minimum degree in each of the
fill graphs ${G^+}^{(t)}$ is at most $\Delta$.  In this case, we maintain $k =
O(\Delta \log n)$ copies of the $\ell_0$-sketch data structure.  By the
coupon-collector argument, any vertex with degree at most $\Delta$ has a list
of all its distinct neighbors with high probability.  This implies that for
each $1 \leq t \leq n$, we can obtain the exact min-degree in ${G^+}^{(t)}$
with high probability.  Figure~\ref{fig:DeltaCappedGlobalVar} gives a brief
description of the data structures we will maintain for this version of the
algorithm.

\begin{figure}[H]
  \begin{algbox}
    \underline{Global Variables}:
    graph $G$ that undergoes pivots,
    degree cap $\Delta$.

    \begin{enumerate}
      \item $k$, the number of copies set to $O(\Delta \log{n})$.
      \item $k$ copies of the $\ell_0$-sketch data structure
      \[
        \variable{DynamicL0Sketch}^{(1)},
        \variable{DynamicL0Sketch}^{(2)},
        \dots,
        \variable{DynamicL0Sketch}^{(k)}.
      \]

      \item For each vertex $u$, a balanced binary search
      tree $\variable{minimizers}(u)$ that stores the minimizers
      of $u$ across all $k$ copies of the data structure.

      \item A balanced binary tree $\variable{bst\_size\_of\_minimizers}$
      on all vertices $u$ with the key of $u$ set to the
      number of different elements in $\variable{minimizers}(u)$.
    \end{enumerate}
  \end{algbox}
  \caption{Global variables for the $\Delta$-capped min-degree algorithm
  \textsc{DeltaCappedMinDegree}.}
  \label{fig:DeltaCappedGlobalVar}
\end{figure}

Note that if we can efficiently maintain the data structures in
Figure~\ref{fig:DeltaCappedGlobalVar}, simply finding the minimum element in
$\variable{bst\_size\_of\_minimizers}$ gives us the vertex with minimum degree.
Theorem~\ref{thm:DataStructureMain} shows that this data structure can
indeed be maintained efficiently.

\begin{theorem}
  \label{thm:DataStructureMain}
  Given i.i.d.\ random variables $x_v$ associated with each vertex $v \in V({G^+}^{(t)})$,
  there is a data structure $\variable{DynamicL0Sketch}$ that, for each vertex
  $u$, maintains the vertex with minimum~$x_v$ among itself and its neighbors
  in ${G^+}^{(t)}$.
  This data structure supports the following methods:
  \begin{itemize}
  \item $\textsc{QueryMin}(u)$, which returns $\xmin(N^{G^{(t)}}_{fill}(u))$
  for a remaining vertex $u$ in $O(1)$ time.
  \item $\textsc{PivotVertex}(u)$, which pivots a remaining vertex $u$
  and returns the list of all remaining vertices $v$ whose values
  of $\xmin(N^{G^{(t)}}_{fill}(v))$ have changed just after this pivot.
  \end{itemize}
  The memory usage of this data structure is $O(m)$.
  Moreover, for any choice of $x$ values for vertices:
  \begin{enumerate}
  \item The total cost of all the pivots is $O(m \log^2{n})$.
  \item 
  For $1 \le t \le n$,
  the total size of all lists returned by $\textsc{PivotVertex}(u^{(t)})$
  is $O(m \log{n})$.
  \end{enumerate}
\end{theorem}

\noindent
This theorem relies on data structures described in Section~\ref{sec:DynamicGraphs},
so we defer the proof to the end of that section.

Now consider a vertex $w$ with fill degree $d \leq \Delta$.
By symmetry of the $x_u$ values, each vertex in $| N_{fill}(w)|$
is the minimizer of $w$ with probability $1 / d$.
As a result, maintaining $O(\Delta \log{n})$ copies of the $\ell_0$-sketch data structure would
ensure that we have an accurate estimation of the minimum fill degree.
The pseudocode for this routine is given in Figure~\ref{fig:DeltaCapped}.
The probability guarantees are formalized in Lemma~\ref{lem:approxdegree},
which is essentially a restatement of \cite[Theorem 2.1]{Cohen97}.

\begin{lemma} \label{lem:approxdegree}
  For a remaining vertex $w$ with fill degree $d \leq \Delta$,
  with high probability we have
  \[
    \variable{bst\_size\_of\_minimizers}[w]=d.
  \]
\end{lemma}

\begin{proof}
The only case where $\variable{bst\_size\_of\_minimizers}[w] \neq d$
is when at least one neighbor of $w$ is not chosen 
in $\variable{minimizers}(w)$.
Let $w'$ be an arbitrary neighbor of $w$ in the fill graph $G^+$.
The probability of $w'$ not being chosen in any of the $k$ copies is
\[
  \left( 1-\dfrac{1}{d} \right)^{k}.
\]
Now, using the assumption that $d \leq \Delta$ and $k = O(\Delta \log n)$, we have
\begin{align*}
  \prob{x_1, x_2, \dots, x_n \sim [0,1)}{w' \text{ not selected in any copy} } & \leq \left( 1-\dfrac{1}{\Delta} \right)^{O(\Delta \log n)}\\
  & \leq e^{-O(\log n)}\\
  &\leq n^{-O(1)}.
\end{align*}

Using a union bound over all neighbors,
we can upper bound the probability that at least one of them is left out by
\[
  \left| N_{fill}(w) \right| \cdot n^{-O(1)} \leq n^{-O(1)},
\]
which completes the proof.
\end{proof}

\begin{figure}[H]

\begin{algbox}
$\textsc{DeltaCappedMinDegree}(G, \Delta)$

\underline{Input}: graph $G=(V,E)$, threshold $\Delta$.

\underline{Output}: exact lexicographically-first min-degree ordering $u^{(1)}, u^{(2)}, \ldots, u^{(n)}$.

\begin{enumerate}
\item For $t=1$ to $|V|$:
  \begin{enumerate}
    \item Set $u^{(t)} \leftarrow \min(\variable{bst\_size\_of\_minimizers})$.
    \item $\textsc{DeltaCappedMinDegree\_Pivot}(u^{(t)})$.
    \end{enumerate}
\end{enumerate}
\end{algbox}

\begin{algbox}
$\textsc{DeltaCappedMinDegree\_Pivot}(u)$

\underline{Input}:
vertex to be pivoted $u$.

\underline{Output}:
updated global state.

\begin{enumerate}

\item For copies $1 \leq i \leq k$:
\begin{enumerate}
\item $\left(v_1, v_2, \ldots, v_{l} \right)
\leftarrow \variable{DynamicL0Sketch}^{(i)}.\textsc{PivotVertex}(u)$,
  the set of vertices in copy $i$ whose minimizers changed after we pivot out $u$.
\item For each $1 \leq j \leq l$:
\begin{enumerate}
\item Update the corresponding values to copy $i$
in $\variable{minimizers}(v_j)$.
\item Update the entry corresponding to $v_j$ in
$\variable{bst\_size\_of\_minimizers}$ with the new
size of $\variable{minimizers}(v_j)$.
\end{enumerate}
\end{enumerate}
\end{enumerate}

\end{algbox}

\caption{Pseudocode for $\Delta$-capped exact min-degree algorithm,
which utilizes the global data structures for \textsc{DeltaCappedMinDegree}
defined in Figure~\ref{fig:DeltaCappedGlobalVar}.}
\label{fig:DeltaCapped}
\end{figure}

\begin{proof}[Proof of Theorem~\ref{thm:BoundedDegreeAlgo}.]
We prove the space bound first.
By Theorem~\ref{thm:DataStructureMain},
each of the $k$ copies of the data structure use $O(m)$ memory.
Each copy of $\variable{minimizers}$ can take space up to
$O(k \log k)$, and $\variable{bst\_size\_of\_minimizers}$ can use
up to $O(n \log n)$ space.
Therefore, total space used is
\[
  O(mk+nk\log k + n\log n) = O(mk) = O(m \Delta \log{n}).
\]

We now analyze the running time.
Theorem~\ref{thm:DataStructureMain} gives a direct cost
of $O(m \log^2{n})$ across all pivots, and in turn
a total cost of $O(m \Delta \log^3{n})$ across all copies.
Furthermore, this implies that the sum of $l$ (the length
of the update lists in $v$) across all steps is at most $O(m \log{n})$.
Each of these updates may lead to one BST update,
so the total overhead is $O(m \log^2{n})$,
which is a lower order term.
\end{proof}

\subsection{Output-Sensitive Running Time}

If we do away with the condition that minimum fill degrees
are bounded above by $\Delta$,
the number of copies of the $\ell_0$-sketch data structure
needed depends on the actual values of the minimum fill degree
at each step.
Therefore, to be more efficient, we can adaptively maintain the required number
of copies of the $\ell_0$-sketch data structure.

For the graph ${G^+}^{(t)}$, we need to have at least
$\Omega(\delta^{(t)} \log n)$ copies of the $\ell_0$-sketch data structure.
However, we do not know the values of $\delta^{(t)}$ a priori.
Therefore, consider the following scheme that adaptively keeps a sufficient
number of copies of the sketch structures:
\begin{enumerate}
	\item Let $C = \delta^{(t-1)}$.
  We will ensure that we have $O(C \log n)$ copies at all times.
  (Note that this is initially true.)
	\item Let $\delta_C$ be the ``computed'' minimum degree
	in ${G^+}^{(t)}$ using $O(C \log n)$ copies of the data structure.
	\item If $\delta_C > C/2$, set $C \leftarrow 2C$ and repeat.
\end{enumerate}

The core idea of the above routine is that if the ``computed'' min-degree is at
most $C/2$, then with high probability the actual min-degree is at most $C$.
Then, because we have $O(C \log n)$ copies of the data structure, the correctness
of the algorithm follows.

\begin{proof}[Proof of Theorem~\ref{thm:OutputSensitiveAlgo}.]
The proof follows analogously to that of Theorem~\ref{thm:BoundedDegreeAlgo},
except that our upper bound for the minimum degrees can be simply given by
$\Delta = 2\cdot \max_{1 \leq t \leq n} \delta^{(t)}$.
With this, the claimed space and time bounds follow. 
\end{proof}

\subsection{Computing Approximate Min-Degree}
\label{subsec:ApproxDegreeDS}
To avoid artificial conditions such as bounds on minimum fill degree
and to make running times independent of the output,
we modify the algorithm to obtain an approximate min-degree vertex at each step.
To do this, we reduce the number of copies of $\variable{DynamicL0Sketch}$
and use the reciprocal of the $(1-1/e)$-th percentile
of a set to approximate its size.\footnote{Note that we use $e$ to refer to the base of the natural logarithm; it should not be confused with edges in a graph.}

However, there is a subtle issue with the randomness that this algorithm uses.
A necessary condition for the algorithm to succeed as intended is that
each step must be independent of its past decisions.
Therefore, we must remove any dependencies between previous and current queries.
Section~\ref{subsec:Correlation} gives an example of such a correlation between
steps of an algorithm.
To circumvent this problem, we need to decorrelate the sketches we construct
and the updates to the data structure from pivoting vertices.
Section~\ref{sec:Decorrelation} tackles this issue.
Rather than simply selecting a vertex with approximate min-degree,
this algorithm requires access to all vertices
whose estimated degree is within a certain range of values.
It follows that this version of the algorithm utilizes such a data structure,
as opposed to the previous two versions which just output the vertex to be
pivoted.

Figure~\ref{fig:ApproxDegreeDSGlobalVar}
gives a description of the data structures for this version of the algorithm.

\begin{figure}[H]
  
  \begin{algbox}
   \underline{Global Variables}: graph $G$,
    error tolerance $\epsilon > 0$.
    
    \begin{enumerate}
      \item Number of copies, $k = O( \log{n} \epsilon^{-2} )$.
 
      \item $k$ copies of the $\ell_0$-sketch data structure
      \[
       \variable{DynamicL0Sketch}^{(1)},
      \variable{DynamicL0Sketch}^{(2)},
      \dots,
      \variable{DynamicL0Sketch}^{(k)}.
      \]

      \item For each vertex $u$, a balanced binary search
      tree $\variable{minimizers}(u)$ that stores the minimizers
      of the $\ell_{0}$-sketch at $u$ across all $k$ copies,
      and maintains the element in $\variable{minimizers}(u)$ with rank
      \[
        \left\lfloor k \left(1- \dfrac{1}{e} \right)\right\rfloor.
      \]

      \item A balanced binary tree $\variable{bst\_quantile}$
      over all vertices $u$ whose key is the
      $\left\lfloor k \left(1-1/e\right)\right\rfloor$-ranked element in $\variable{minimizers}(u)$.
      
    \end{enumerate}
  \end{algbox}

\caption{Global variables and data structures for $\textsc{ApproxDegreeDS}$,
which returns (implicit) partitions of vertices into buckets with $\epsilon$-approximate degrees.}
\label{fig:ApproxDegreeDSGlobalVar}
\end{figure}

To achieve our goal of using fewer copies of the data structure, we use a
sampling-based algorithm.  In particular, we make use of the following lemma.

\begin{restatable}{lemma}{minValue}
\label{lem:minValue}
Suppose that we have $k$ copies of the $\ell_0$-sketch data structure
for some $k \geq \Omega(\log{n} \epsilon^{-2})$.
Let $w$ be a vertex with degree $d > k$,
and let $q(w)$ denote the
$\lfloor k \left(1-1/e\right)\rfloor$-ranked
element in $\variable{minimizers}(w)$.
Then, with high probability, we have
\[
\frac{1 - \epsilon}{d} \leq q(w) \leq \frac{1 + \epsilon}{d}.
\]
\end{restatable}

\noindent
Lemma~\ref{lem:minValue} is simply a restatement of 
\cite[Propositions 7.1 and 7.2]{Cohen97}.
However, \cite{Cohen97} assumes that the random variables
$x_u$ are drawn from the exponential distribution (and hence also their minimum),
whereas we assume that $x_u$ is independently drawn from the uniform distribution.
When $d$ is large though, the minimums of $d$
elements from both distributions are almost identically distributed.  For the
sake of completeness, we provide the proof for the $x_u$ variables being
uniformly distributed in Appendix~\ref{sec:SketchingProofs}.

This leads to the following result for providing implicit access to all
vertices with approximately the same degree, which is crucial for our overall
nearly-linear time algorithm in Section~\ref{sec:Decorrelation}.
We give its pseudocode in Figure~\ref{fig:approxdegree}.

\begin{figure}[H]
  
  \begin{algbox}
    
    $\textsc{ApproxDegreeDS\_Pivot}(u)$
    
    \underline{Input}: vertex to be pivoted, $u$.
    
    \underline{Output}: updated global state.
    
    \begin{enumerate}
      \item For each copy $1\leq i \leq k$:
         \begin{enumerate}
            \item $\left(v_1, v_2, \ldots, v_{l} \right)
            \leftarrow \variable{DynamicL0Sketch}^{(i)}.\textsc{PivotVertex}(u)$,
             the set of vertices in copy $i$ whose minimizers changed after we pivot out $u$.
            \item For each $1 \leq j \leq l$:
              \begin{enumerate}
                \item Update the value of corresponding to copy $i$ in $\variable{minimizers}(v_j)$,
which in turn updates its $\lfloor k(1 - 1/e) \rfloor$-ranked quantile.
                \item Update the entry corresponding to $v_j$ in
$\variable{bst\_quantile}$ with the new value of the
$\lfloor k(1 - 1/e) \rfloor$-ranked quantile of $\variable{minimizers}(v_j)$.
\end{enumerate}
\end{enumerate}
    \end{enumerate}
    
  \end{algbox}
  
  \begin{algbox}
    
    $\textsc{ApproxDegreeDS\_Report()}$
    
    \underline{Output}: approximate bucketing of the vertices
    by their fill-degrees.
    
    \begin{enumerate}
      \item For each $i$ from $0$ to $B = O(\log{n} \epsilon^{-1})$:
      \begin{enumerate}
        \item Set $S_i$ to be the split binary tree
        from $\variable{bst\_quantile}$ that contains all
        nodes with $\lfloor k(1 - 1/e) \rfloor$-ranked quantiles in the range
      \[
        \left[\left(1 + \epsilon\right)^{-i - 1}, \left(1 + \epsilon\right)^{-i}\right].
      \]
      \end{enumerate}
     \item Return the tuple
       $(S_1, S_2, \dots, S_{B})$.
    \end{enumerate}
    
  \end{algbox}
  
  \caption{Pseudocode for data structure that returns pointers
  to binary trees containing partitions of remaining vertices into
  sets with $\epsilon$-approximate degrees.
Its corresponding global variables are defined in Figure~\ref{fig:ApproxDegreeDSGlobalVar}.}
  \label{fig:approxdegree}
\end{figure}

%

When interacting with $\textsc{ApproxDegreeDS\_Report}()$,
note that the maximum degree is $O(n)$,
so we have $B = O(\log{n} \epsilon^{-1})$.
Therefore, the data structure can simply return pointers
to ``samplers'' for the partition $V_1, V_2, \ldots, V_B$.

\begin{proof}[Proof of Theorem~\ref{thm:ApproxDegreeDS}.]
By construction,
all vertices in $V_i$ have their $\lfloor k(1 - 1/e) \rfloor$-ranked quantile
in the range 
\[
\left[(1+\epsilon)^{-i-1},(1+\epsilon)^{-i}\right].
\]
Subsequently from Lemma~\ref{lem:minValue},
the fill-degree of a vertex $w \in V_i$
is in the range \[
\left[(1-\epsilon)(1+\epsilon)^{i}, (1+\epsilon)(1+\epsilon)^{i+1}\right],
\]
which is within the claimed range for $\epsilon \leq 1/2$.

The proof of time and space bounds is again analogous to that of
Theorem~\ref{thm:BoundedDegreeAlgo}.
Substituting in the new number of copies $k = O(\log{n} \epsilon^{-2})$ instead
of $\Delta$ proves the space complexity.

The main difference in these data structures is that
we now need to store information about the
$\lfloor k(1 - 1/e) \rfloor$-ranked quantile.
These can be supported in $O(\log{n})$ time by augmenting the balanced binary
search trees with information about sizes of the subtrees
in standard ways (e.g. \cite[Chapter 14]{CLRS3}).
A $O(\log{n})$ time splitting operation is also standard to
  most binary search tree data structures (e.g. treaps~\cite{Seidel96}).
\end{proof}

\noindent
Note that there may be some overlaps between the allowed ranges of the buckets;
vertices on the boundary of the buckets may be a bit ambiguous.

An immediate corollary of Theorem~\ref{thm:ApproxDegreeDS} is that
we can provide access to approximate minimum-degree vertices
for a fixed sequence of updates by always returning
some entry from the first non-empty bucket.
\begin{corollary}
\label{cor:ApproxMindegreeOblivious}
For a fixed sequence of pivots $u^{(1)}, u^{(2)}, \dots, u^{(n)}$,
we can find $(1 + \epsilon)$-approx min-degree vertices
in each of the intermediate states in 
$O(m \log^{3}n\epsilon^{-2})$ time.
\end{corollary}


\section{Generating Decorrelated Sequences}
\label{sec:Decorrelation}

In this section we show our nearly-linear
$(1 + \epsilon)$-approximate min degree algorithm.
The algorithm crucially uses the \textsc{ApproxDegreeDS}
data structure constructed in Section~\ref{subsec:ApproxDegreeDS}.

\begin{theorem}
\label{thm:ApproxMinDegree}
There is an algorithm $\textsc{ApproxMinDegreeSequence}$
that produces an $\epsilon$-approximate greedy min-degree sequence
in expected $O(m \log^5 n \epsilon^{-2})$ time with high probability.
\end{theorem}

The algorithm is based on the degree approximation
routines using sketching, as described
in Theorem~\ref{thm:ApproxDegreeDS}.
Theorem~\ref{thm:ApproxDegreeDS} provides us access
to vertex buckets, where the $i$-th bucket contains vertices
with fill degrees in the range $[(1+\epsilon)^{i-2},(1+\epsilon)^{i+2}]$.
At any point, reporting any member of the first non-empty bucket gives an
approximate minimum degree choice.
However, such a choice must not have dependencies on the randomness used to
generate this step, or more importantly, subsequent steps.

To address this issue, we use an additional layer of randomization,
which decorrelates the $\ell_0$-sketch data structures
and the choice of vertex pivots.
Figure~\ref{fig:ApproxMinDegreeSequence} contains the pseudocode
for the top-level algorithm to compute a nearly-linear
$(1 + \epsilon)$-approximate minimum degree sequence.
The algorithm makes calls the following routines and data structures:
\begin{itemize}
	\item \textsc{ApproxDegreeDS}: Access to buckets of vertices
	with approximately equal degrees (Section~\ref{subsec:ApproxDegreeDS}).
	\item \textsc{ExpDecayedCandidates}: Takes a set
	whose values are within $1 \pm \epsilon$ of each other,
	randomly perturbs its elements, and returns this ($\epsilon$-decayed) set.
	\item \textsc{EstimateDegree}: Gives an $\epsilon$-approximation
	to the fill-degree of any given vertex (Section~\ref{sec:DegreeEstimation}).
	The formal statement is given in Theorem~\ref{thm:DegreeEstimation}.
\end{itemize}

\begin{theorem}
	\label{thm:DegreeEstimation}
	There is a data structure that maintains a component graph $G^{\circ}$
	under (adversarial) vertex pivots in a total of $O(m \log^2{n})$ time,
	and supports the operation $\textsc{EstimateDegree}(G^{\circ}, u, \epsilon)$,
	which given a vertex $u$ and
	error threshold $\epsilon > 0$,
	returns with high probability an $\epsilon$-approximation to the fill-degree of
	$u$ by making $O(d_{component}^{G^{\circ}} (u) \log^2{n} \epsilon^{-2})$ oracle
	queries to $G^{\circ}$.
\end{theorem}

\begin{figure}[H]
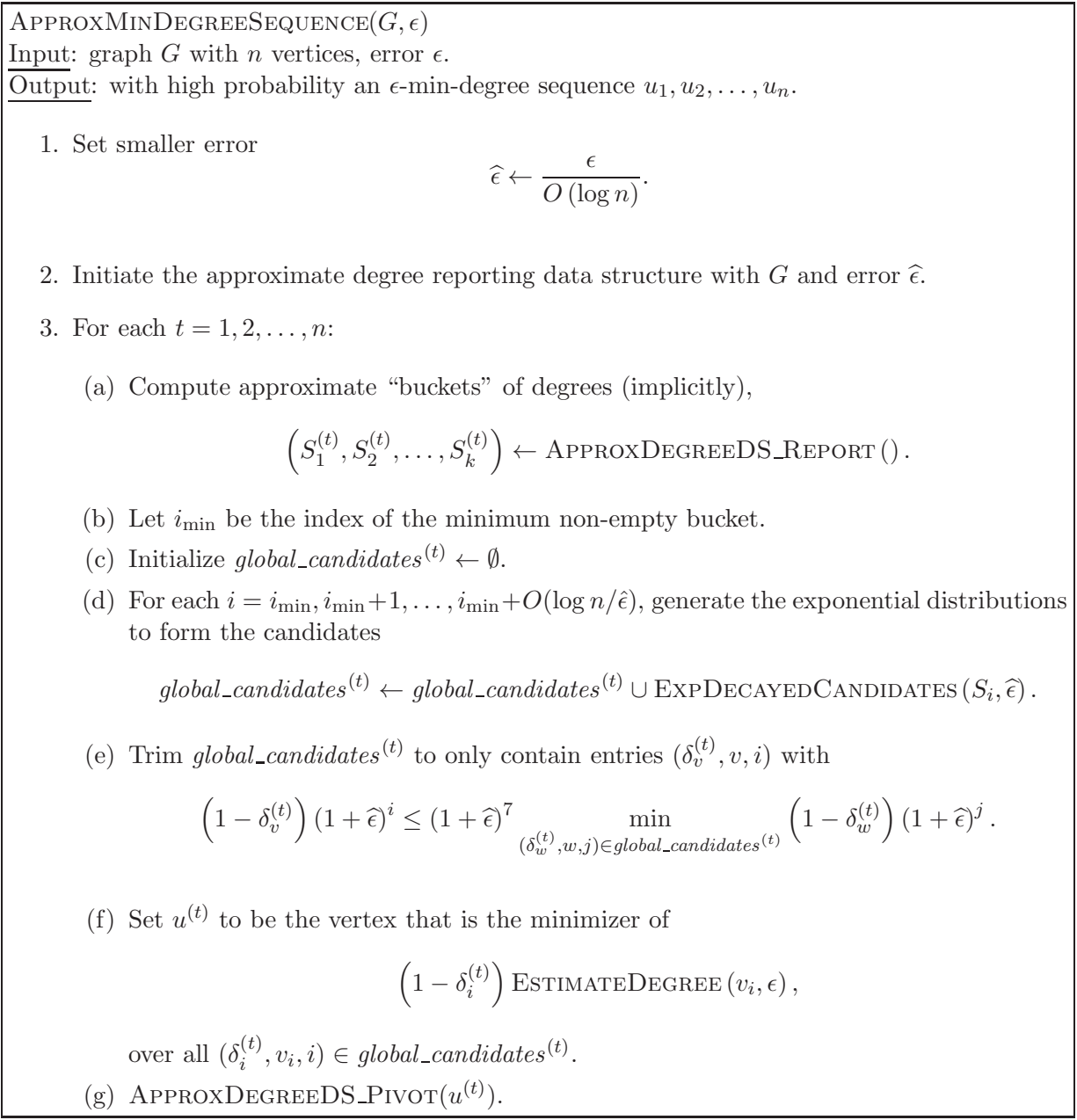

	
	\begin{algbox}
		
		$\textsc{ApproxMinDegreeSequence}(G, \epsilon)$
		
		\underline{Input}: graph $G$ with $n$ vertices, error $\epsilon$.
		
		\underline{Output}: with high probability an $\epsilon$-min-degree
		sequence $u_1, u_2, \dots, u_n$.
		
		\begin{enumerate}
			
			\item Set smaller error
			\[
			\widehat{\epsilon}
			\leftarrow
			\frac{\epsilon}{O\left(\log{n}\right)}.
			\] \label{algline:smallerror}
			
			\item 
			Initiate the approximate degree reporting data structure with $G$
			and error $\widehat{\epsilon}$.
			
			\item For each $t = 1, 2, \dots, n$:
			
			\begin{enumerate}
				
				\item Compute approximate ``buckets'' of degrees (implicitly),
				\[
				\left( S^{(t)}_1 , S^{(t)}_2, \dots, S^{(t)}_k \right)
				\leftarrow
				\textsc{ApproxDegreeDS\_Report}\left( \right).
				\]
				
				\item Let $i_{\min}$ be the index of the minimum non-empty bucket.
				
				\item Initialize $\variable{global\_candidates}^{(t)} \leftarrow \emptyset$.
				
				\item For each $i=i_{\min}, i_{\min} + 1,\dots, i_{\min} + O( \log{n}/\hat{\epsilon})$,
				generate the exponential distributions to form the candidates
				\[
				\variable{global\_candidates}^{(t)}
				\leftarrow \variable{global\_candidates}^{(t)} \cup\\ 
				\textsc{ExpDecayedCandidates}\left( S_{i}, \widehat{\epsilon} \right).
				\]
				
				\item Trim $\variable{global\_candidates}^{(t)}$ to only contain
				entries $(\delta_v^{(t)}, v, i)$ with
				\[
				\left( 1 - \delta^{\left(t\right)}_{v} \right)
				\left( 1 + \widehat{\epsilon} \right)^{i}
				\leq
				\left( 1 + \widehat{\epsilon} \right)^{7}
				\min_{(\delta_w^{(t)}, w, j) \in \variable{global\_candidates}^{(t)}}
				\left( 1 - \delta^{\left(t\right)}_{w} \right)
				\left( 1 + \widehat{\epsilon} \right)^{j}.
				\] \label{algline:trimGlobalCandidates}
				\item Set $u^{(t)}$ to be the vertex that is the minimizer of
				\[
				\left( 1 - \delta^{\left(t\right)}_{i} \right)
				\textsc{EstimateDegree}\left( v_{i}, \epsilon \right),
				\]
				over all $(\delta^{\left(t\right)}_{i}, v_{i}, i) \in \variable{global\_candidates}^{(t)}$.
				\item $\textsc{ApproxDegreeDS\_Pivot}(u^{(t)})$.
			\end{enumerate}
			
		\end{enumerate}
		
	\end{algbox}
	
	\caption{Pseudocode for approximate min-degree algorithm.}
	\label{fig:ApproxMinDegreeSequence}
\end{figure}

The most important part of the algorithm is arguably the use of exponential
distributions to form candidates in a way that it is completely decorrelated
with the randomness used to generate the $\ell_0$-sketch data structure and in
the choice of previous vertex pivots.  The following subsection summarizes some
of the desirable properties of exponential random variables that we exploit in
our algorithm.

\subsection{Exponential Random Variables}
The exponential distribution has been well-studied.
In particular, we use properties of its order statistics, which
arise in the study of fault tolerance and
distributed graph decomposition~\cite{MillerPX13}.
For a parameter $\lambda$, this distribution
is defined by the probability density function (PDF)
\begin{align*}
  f_{\Exp}(x; \lambda) =
  \begin{cases}
    \lambda \exp\left(-\lambda x\right) & \text{if $x \ge 0$,}\\
    0 & \text{otherwise.}
  \end{cases}
\end{align*}
We denote this distribution by $\Exp(\lambda)$,
and also make use of its cumulative density function (CDF)
\begin{align*}
  F_{\Exp}(x; \lambda) &=
  \begin{cases}
    1 - \exp\left(-\lambda x\right) & \text{if $x \ge 0$,}\\
    0 & \text{otherwise}.
  \end{cases}
\end{align*}

A crucial fact about the exponential distribution is that it is memoryless.
That is, if we condition on $\Exp(\lambda) \geq t$,
then $\Exp(\lambda) - t$ follows the same distribution.
A substantial part of our analysis relies on the \emph{order statistics} of
exponential random variables.
Given $n$ random variables $X_1, X_2, \dots, X_n$, the $i$-th
order statistic is  the value of the $i$-th minimum random variable.
A useful property of exponential distributions is that the difference
between its order statistics also follows an exponential distribution,
which we exploit when sampling exponential random variables in decreasing order.

\begin{lemma}[{\cite{Feller71:book}}]
\label{lem:OrderStatisticExp}
Let $X_{(i)}^n$ denote the $i$-th order statistic of $n$
i.i.d.\ random variables drawn from the distribution $\Exp(\lambda)$.
Then, the $n$ variables $X_{(1)}, X_{(2)}- X_{(1)},\dots, X_{(n)} - X_{(n - 1)}$
are independent, and the density of $X_{(k + 1)} - X_{(k)}$
is given by the distribution $\Exp((n - k)\lambda)$.
\end{lemma}

One approach to prove Lemma~\ref{lem:OrderStatisticExp} uses the i.i.d.\
assumption to show that the CDF of
$X_{(1)}^n$ is 
\begin{align*}
  F_{X^n_{(1)}}(x)&=1-(1-F(x))^n
  \\&=1-\exp(-n\lambda x),
\end{align*}
where $F(x)=1-\exp(-\lambda x)$ is the CDF of
$\Exp(\lambda)$.
This proves that $X^n_{(1)}$ follows an exponential
distribution with mean $1/(n\lambda)$.
Then conditioning on $X^n_{(1)}$, we see that $X^n_{(2)}-X^n_{(1)}$
again follows an exponential distribution equal to $X_{(1)}^{n-1}$
because of the memoryless property. We can repeat this argument to get the
density of $X^n_{(k+1)}-X^n_{(k)}$ for all $k$ up to $n-1$.

The key definition in this section is a sequence defined by exponential
perturbations.  It is motivated by Theorem~\ref{thm:ApproxDegreeDS}, which
states that all the vertices are grouped approximately by degrees.  In the
following definition, $n$ is global and equal to the original number of
vertices in the graph.
Also, we let $c_1 > 1$ be some fixed constant.

\begin{definition}
\label{def:DecayedMinimum}
Given a set of values $\{x_1,x_2, \dots, x_k\}$,
an \emph{$\epsilon$-decayed minimum} of this set is generated
by independently drawing the exponential random variables
\[
  \delta_{i} \sim \hat{\epsilon} \cdot \Exp(1),
\]
where $\hat{\epsilon} = \epsilon/(c_1 \cdot \log n)$
(line~\ref{algline:smallerror} in \textsc{ApproxMinDegreeSequence}),
and returning 
\[
  \min_i \left(1 - \delta_i \right) x_i.
\]
\end{definition}

\begin{definition}
\label{def:PerturbedMinDegreeSequence}
Given a parameter $\epsilon$ and an $\epsilon$-degree estimation
routine $\textsc{EstimateDegree}(G, \cdot)$,
we define an \emph{$\epsilon$-decayed min-degree sequence} as a
sequence such that:
\begin{enumerate}
\item The next vertex to be pivoted, $u^{(t)}$,
is the one corresponding to the $\epsilon$-decayed minimum of the values
\[
  \textsc{EstimateDegree}\left(G^{\left(t - 1 \right)}, v \right),
\]
over all remaining vertices $v$ in $G^{(t - 1)}$.
\item $G^{(t)}$ is the graph obtained after pivoting $u^{(t)}$ from $G^{(t-1)}$.
\end{enumerate}
\end{definition}

Importantly, note that the randomness of this degree estimator is regenerated
at each step, thus removing all dependencies.
Section~\ref{subsec:epsilonDecayed} describes how to generate
this distribution implicitly.
We first show that this approximation is well-behaved.
\begin{lemma}
\label{lem:ApproxFactor}
Let $Y$ be an $\epsilon$-decayed minimum
of $\{x_1, x_2, \dots, x_k\}$.
With high probability, we have
\[
  Y \ge (1 - \epsilon) \min\{x_1, x_2, \dots, x_k\}.
\]
\end{lemma}

\begin{proof}
%
Let us bound the probability of the complementary event
\[
  Y < (1 - \epsilon ) \min\{x_1, x_2, \dots, x_k\}. 
\]
Observe that we can upper bound this probability by
the probability that some $x_i$ decreases
to less than $1 - \epsilon$ times its original value.
Recall that we set $\hat{\epsilon} = \epsilon/(c_1 \cdot \epsilon)$
for some constant $c_1 > 1$ (Line~\ref{algline:smallerror} in \textsc{ApproxMinDegreeSequence}).
Consider $k$ i.i.d.\ exponential random variables
\[
  X_1,X_2,\dots,X_k \sim \Exp(1),
\]
and let
\[
  \delta_i = \hat{\epsilon} \cdot X_i,
\]
as in the definition of an $\epsilon$-decayed minimum
(Definition~\ref{def:DecayedMinimum}).
By the CDF
of the exponential function, for each $1 \le i \le k$, we have
\begin{align*}
  \prob{\delta_i} {\delta_{i} > \epsilon }
  &=
  \prob{X_i} {\frac{\epsilon}{c_1 \log n} \cdot X_i > \epsilon}\\
  &= \prob{X_i} {X_i > c_1 \log n}\\
  &= \exp \left( - c_1 \log{n} \right).
\end{align*}
By a union bound, it follows that
\begin{align*}
\prob{\delta_{1} \dots \delta_{k}}
  {\max_{1 \le i \le k}
     \delta_{i} >  \epsilon }
  &= \prob{\delta_{1} \ldots \delta_{k}}
  {\text{there exists } 1 \le i \le k \text{ such that } 
    \delta_{i} > \epsilon}\\
  &\leq
\sum_{1 \le i \le n}
  \prob{\delta_{1}}{\delta_{1} > \epsilon }.
\end{align*}
Substituting in the bound from the CDF for each $\delta_{i}$ gives
\begin{align*}
\prob{\delta_{1} \ldots \delta_{k}}
  {\max_{1 \le i \le k}
     \delta_{i} >  \epsilon }
  &\le n \cdot \exp\left( - c_1 \log{n} \right)\\
  &= n^{1 - c_1}.
\end{align*}
Considering the complementary event gives the result for $c_1 > 1$.
\end{proof}

The above lemma implies that, to produce an
$\epsilon$-approximate greedy minimum degree sequence
(as defined in Definition~\ref{def:ApproxMinDegree}),
it suffices to compute an
$\epsilon$-decayed minimum-degree sequence.
Specifically, at each iteration, we only need to find the
$\epsilon$-decayed minimum among the (approximate)
degrees of the remaining vertices.

It turns out, however, that computing the approximate
degrees for each remaining vertex during every iteration is rather expensive.
Section~\ref{subsec:epsilonDecayed} shows how we can tackle
this problem by carefully choosing a candidate subset of vertices
at each iteration.

\subsection{Implicitly Generating the $\epsilon$-Decayed Minimum}
\label{subsec:epsilonDecayed}

We now consider the problem of finding the $\epsilon$-decayed
minimum of a set of vertex degrees.
Since the number of these vertices can be huge,
our first step is to find a small candidate subset.
This is done via the routine \textsc{ExpDecayedCandidates},
and its pseudocode is given in Figure~\ref{fig:ExpDecayedCandidates}.
Once again, let $\hat{\epsilon}$ denote $\epsilon/(c_1 \cdot \log n)$,
for some constant $c_1 > 1$.

\begin{figure}[H]
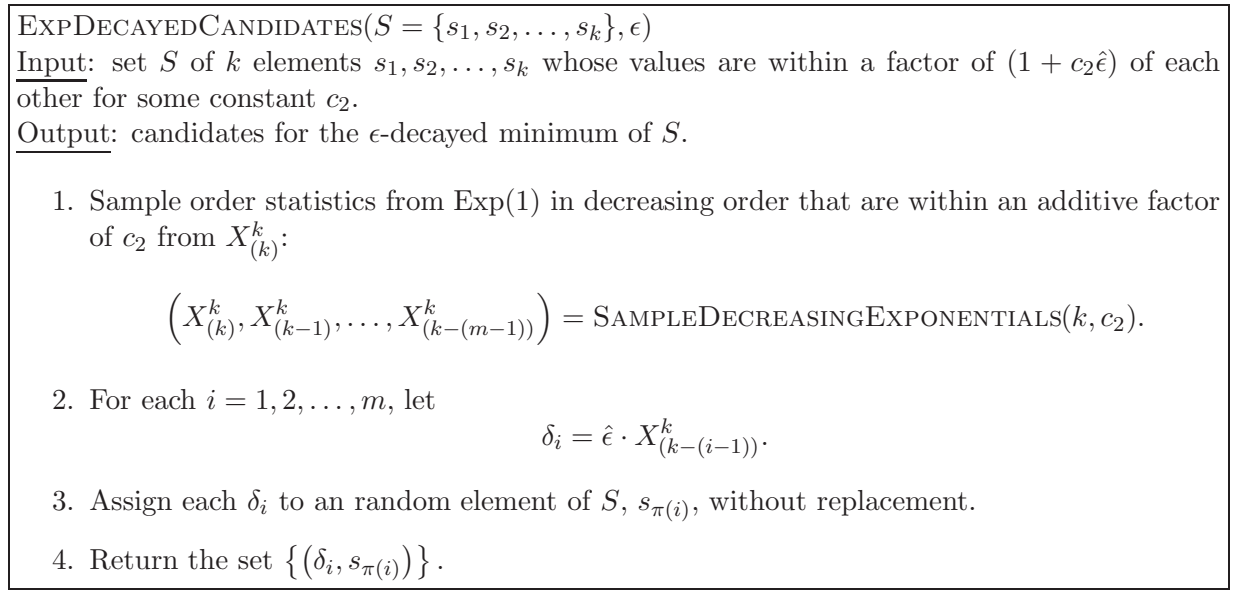


\begin{algbox}
$\textsc{ExpDecayedCandidates}
  (S = \{s_1, s_2,\dots, s_{k}\}, \epsilon)$

\underline{Input}:
set $S$ of $k$ elements $s_{1},s_2, \dots, s_k$
whose values are within a factor of $(1 + c_2 \hat{\epsilon})$ of each other
for some constant $c_2$.

\underline{Output}:
candidates for the $\epsilon$-decayed minimum of $S$.

\begin{enumerate}

\item 
Sample order statistics from $\Exp(1)$ in decreasing order
that are within an additive factor of $c_2$ from $X_{(k)}^k$:
\[
  \left( X_{(k)}^k, X_{(k-1)}^k, \dots, X_{(k - (m - 1))}^k \right) =
    \textsc{SampleDecreasingExponentials}(k, c_2).
\]

\item For each $i = 1,2,\dots,m$, let
\[
  \delta_{i} = \hat{\epsilon} \cdot X_{(k - (i - 1))}^k.
\]

\item Assign each $\delta_{i}$ to an random element
of $S$, $s_{\pi(i)}$, without replacement.

\item Return the set
$
\left\{
  \left(\delta_{i}, s_{\pi\left(i\right)}\right)
\right\}.
$
\end{enumerate}

\end{algbox}

\caption{
Pseudocode for returning an $O(1)$-sized candidate subset
for the $\epsilon$-decayed minimum of a given set
of values that are within  $(1 + c_2\hat{\epsilon})$
of each other.}

\label{fig:ExpDecayedCandidates}

\end{figure}

Notice that the input requires all the elements
to be within a factor of $(1 + c_2\hat{\epsilon})$
of each other.
The way we achieve this is simply using the
vertex buckets produced by our algorithm
\textsc{ApproxDegreeDS} in Section~\ref{subsec:ApproxDegreeDS},
using $\hat{\epsilon}$ as the tolerance value.
The following lemma shows that approximate vertex degrees in
one such bucket satisfies the required condition on the input.

\begin{lemma}\label{lem:BucketBounds}
	For an arbitrary bucket $B_i^{(t)}$,
	there exists a constant $c_2$ such that
	all its approximate degree values are within
	a factor of $(1+c_2 \hat{\epsilon})$
	(or alternatively, within a factor of $(1+\hat{\epsilon})^7$)
	of each other.
\end{lemma}

\begin{proof}
	From Theorem~\ref{thm:ApproxDegreeDS},
	the $i$-th bucket has vertices with
	degrees in the range
	\[
	\left[(1+\hat{\epsilon})^{i-2},(1+\hat{\epsilon})^{i+2}\right].
	\]
	From Theorem~\ref{thm:DegreeEstimation}, we have oracle access to the graph
	${G^\circ}^{(t)}$ and can therefore invoke \textsc{EstimateDegree} on it.
	Instead of treating each call to $\textsc{EstimateDegree}$
	and the values of $\delta_u$ used to generate the $\epsilon$-decayed
	minimum as random variables, we consider them as fixed values
	after removing the bad cases using w.h.p.
	That is, we define
	\[
	\tilde{d}_{fill}^{\left( t \right)}\left(u \right)
	\defeq
	\textsc{EstimateDegree}\left(G^{\circ(t-1)}, u, \hat{\epsilon} \right).
	\]
	By Theorem~\ref{thm:DegreeEstimation}, with high probability,
	every call to $\textsc{EstimateDegree}$ is correct,
	so we have
	\begin{equation*}
	\left( 1 - \epsilon \right)
	d_{fill}^{G^{\circ \left( t-1 \right)} }\left( u \right)
	\leq
	\tilde{d}_{fill}^{\left( t \right)}\left(u \right)
	\leq
	\left( 1 + \epsilon \right)
	d_{fill}^{G^{\circ \left( t-1 \right)} }\left( u \right).
	\label{eq:ApproxDegree}
	\end{equation*}
	This implies that the values are in the range
	\[
	\left[(1+\hat{\epsilon})^{i-4},(1+\hat{\epsilon})^{i+3}\right].
	\]
	Hence, all values in a bucket are within a factor of
	\[
	\dfrac{ (1+\hat{\epsilon})^{i+3} }{ (1+\hat{\epsilon})^{i-4} }
	= (1+\hat{\epsilon})^{7} \leq (1 + c_2\hat{\epsilon})
	\]
	of each other.
\end{proof}

Recall that $X_{(i)}^k$ denotes the distribution
of the $i$-th smallest value among $k$ identically
sampled variables.
The most important part of the above algorithm
is sampling the required order statistics efficiently.
Figure~\ref{fig:SamplingDecreasingExponentials}
shows how to sample the order statistics
\[
X_{(k)}^k, X_{(k - 1)}^k, \dots, X_{(1)}^k
\]
from $\Exp(1)$
in decreasing order iteratively
using Lemma~\ref{lem:OrderStatisticExp}.

\begin{figure}[H]
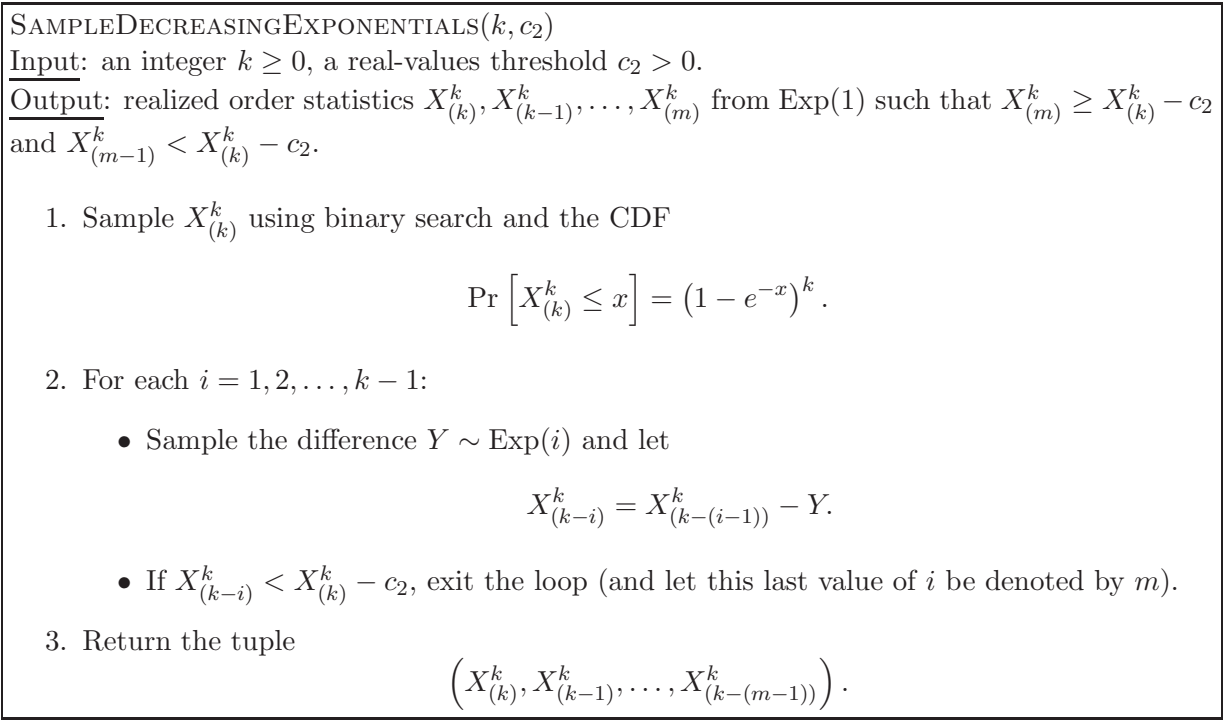

\begin{algbox}
\textsc{SampleDecreasingExponentials}$(k, c_2)$

\underline{Input}:
an integer $k \ge 0$, a real-values threshold $c_2 > 0$.

\underline{Output}:
realized order statistics $X_{(k)}^k, X_{(k-1)}^k, \dots, X_{(m)}^k$
from $\Exp(1)$
such that $X_{(m)}^k \ge X_{(k)}^k - c_2$
and $X_{(m-1)}^k < X_{(k)}^k - c_2$.

\begin{enumerate}

\item Sample $X_{(k)}^k$ using binary search and the CDF
\[
    \prob{}{X_{(k)}^k \le x} =
    \left(1 - e^{-x}\right)^k.
\]

\item For each $i = 1, 2, \dots, k-1$:
\begin{itemize}
  \item Sample the difference $Y \sim \Exp(i)$ and let
    \[
      X_{(k-i)}^k = X_{(k - (i - 1))}^k - Y.
    \]
  \item If $X_{(k-i)}^k < X_{(k)}^k - c_2$, exit the loop (and let this last value of $i$ be denoted by $m$).
\end{itemize}
  \item Return the tuple
    \[
      \left(X_{(k)}^k, X_{(k-1)}^k, \dots, X_{(k-(m-1))}^k\right).
    \]
\end{enumerate}

\end{algbox}

\caption{
Pseudocode for iteratively generating order statistics of exponential random variables
in decreasing order within a threshold $c_2$ of the maximum value $X_{(k)}^k$.
}

\label{fig:SamplingDecreasingExponentials}

\end{figure}

To show that our algorithm is correct,
we must prove two parts: that the algorithm
picks $O(1)$ candidates in expectation,
and with high probability the actual
$\epsilon$-decayed minimum belongs to this candidate set.

\begin{lemma}
\label{lem:FindingDecayedMin}
Suppose $x_1, x_2, \dots, x_k$ are within a factor of
$(1 + c_2\hat{\epsilon})$ of each other.
Then the $\epsilon$-decayed minimum of this set
is among the candidates returned by
$\textsc{ExpDecayedCandidates}(\{x_1, x_2, \dots, x_k\}, \epsilon)$.
Furthermore, the expected number of candidates returned is $O(1)$.
\end{lemma}

\begin{proof}
Let $X_1,X_2, \dots, X_k \sim \Exp(1)$ be i.i.d.,
and the order statistic $X_{(k)}^{k} = \max\{X_1,X_2,\dots,X_k\}$.
First let us verify the correctness of
\textsc{SampleDecreasingExponentials}.
Observe that the CDF of $X_{(k)}^k$ is
\begin{align*}
  F_{X_{(k)}^k}(x) &= \prob{}{\max\{X_1,X_2,\dots,X_n\} \le x}\\
  &= \prod_{i=1}^k \prob{}{X_i \le x}\\
  &= \left(1 - e^{-x}\right)^k.
\end{align*}
Therefore, $X_{(k)}^k$ is sampled correctly in the algorithm.
Using the memoryless property of the exponential distribution in
Lemma~\ref{lem:OrderStatisticExp}, we can
iteratively generate
$X_{(k-1)}^k, X_{(k-2)}^k, \dots, X_{(1)}^k$ by sampling their differences
from varying exponential distributions.
Let $i^* = \argmin_{i} 1 - \delta_i$. 
Now, to ensure that the $\epsilon$-decayed minimum
is not left out from our candidate set,
we need to sample every
$\delta_j$ such that
\[ 
  \delta_j \ge \delta_{i^*} - c_2 \hat{\epsilon}.
\]
To verify that this suffices,
suppose that the $\epsilon$-decayed minimum (say $x_{\pi(\ell)}$)
is not included in our candidate set.
Then,
\[
(1-\delta_{i^*})x_{\pi(i^*)} \leq (1-\delta_{i^*}+c_2 \hat{\epsilon})(x_{i^*}- c_2 \hat{\epsilon}) \leq (1-\delta_{\ell})x_{\pi(\ell)},
\]
which is a contradiction.

Lastly, to count the number of elements included in the
candidate set,
we count the number of such $\delta_j$ samples.
Equivalently,
we bound the number of random variables $X_j \sim \Exp(1)$
such that
\[
X_j \ge X_{(k)}^k - c_2,
\]
using Definition~\ref{def:DecayedMinimum}.
Let $Z_i$ be the indicator variable which equals $1$ when $X_j \geq X_{(k)}^k - c_2$,
and let $Z = \sum_{1 \le i \le k} Z_i$ indicate the size of our candidate subset.
Using the memoryless property of exponentials,
\begin{align*}
  \E[Z] &= \sum_{i=1}^k \E[Z_i]\\
     &= \sum_{i=1}^k \prob{}{X_{(i)}^k \ge X_{(k)}^k - c_2} \\
     &= 1 + \sum_{i=1}^{k-1} \prob{}{X_{(i)}^k \ge X_{(k)}^k - c_2} \\
     &= 1 + \sum_{i=1}^{k-1} \prob{}{X_{(i)}^i \le c_2}\\
     &= 1 + \sum_{i=1}^{k-1} \left( 1 - e^{-c_2} \right)^i\\
     &\le e^{c_2},
\end{align*}
where the final equality sums the geometric series.
Therefore, $O(1)$ exponential random variables are generated as we sample
backwards from the maximum.
\end{proof}

Note that we cannot simply work with the smallest bucket,
because the randomness introduces a $1 \pm \epsilon$
perturbation.
Even the bucket location of the vertex of minimum degree
is dependent on the randomness of the sketches used
to generate them (discussed in Theorem~\ref{thm:ApproxDegreeDS}).
So, the algorithm finds this $O(1)$-sized candidate set
for $O(\log n \hat{\epsilon}^{-1})$ buckets,
which suffices since the $\epsilon$-decayed minimum
cannot be in the latter buckets.
However, this increases our global candidate set to
size $O(\log n \hat{\epsilon}^{-1})$.
As a final step before computing degrees of vertices,
we show that we can trim this set carefully, while still
keeping the $\epsilon$-decayed minimum in it
with high probability.

\begin{lemma}\label{lem:trimmingisokay}
	Let $(\delta_{v}^{(i)},v,i)$ be the entry that corresponds to
	the $\epsilon$-decayed minimum in the set
	$\variable{global\_candidates}^{(t)}$. Then with high probability,
  we have
	\[
	\left( 1 - \delta^{\left(t\right)}_{v} \right)
	\left( 1 + \widehat{\epsilon} \right)^{i}
	\leq
	\left( 1 + \widehat{\epsilon} \right)^{7}
	\min_{(\delta_w^{(t)}, w, j) \in \variable{global\_candidates}^{(t)}}
	\left( 1 - \delta^{\left(t\right)}_{w} \right)
	\left( 1 + \widehat{\epsilon} \right)^{j}.
	\]
\end{lemma}

\begin{proof}
	Let $(\delta_{u}^{(j)},u,j)$ be an arbitrary entry in the set
	$\variable{global\_candidates}^{(t)}$.
	We know that
	\[
	\left( 1 - \delta^{\left(t\right)}_{v} \right) d_{fill}(v)
	\leq \left( 1 - \delta^{\left(t\right)}_{u} \right) d_{fill}(u).
	\]
	From Lemma~\ref{lem:BucketBounds},
	\[
	d_{fill}(v) \leq (1+ \hat{\epsilon})^{i+3}
	\]
	and
	\[
	d_{fill}(u) \geq (1+ \hat{\epsilon})^{j-4}.
	\]
	Substituting these into the previous inequality
	gives us the result.
\end{proof}

Now we use all our building blocks from this section to prove the correctness
of the algorithm.

\begin{lemma}
\label{lem:Correctness}
For any graph $G$ and any error $\epsilon$,
the output of $\textsc{ApproxMinDegreeSequence}(G, \epsilon)$
is with high probability
an $\epsilon$-approximate greedy min-degree sequence.
\end{lemma}

\begin{proof}
We prove by induction that for some constant $c$, we can show
that after $t$ steps, our sequence is an $\epsilon$-approximate
greedy min-degree sequence with probability at least $1 - t \cdot n^{-c}$.
The base case of $t = 0$ follows because nothing has happened so far.
As the inductive hypothesis, suppose we have an $\epsilon$-approximate
greedy min-degree sequence $u^{(1)}, u^{(2)}, \dots, u^{(t)}$.
Then consider the graph $G^{(t)}$ where these vertices
are marked as eliminated and the rest of the vertices 
are marked as remaining.

From Lemma~\ref{lem:BucketBounds},
all values in a bucket are within a factor
of $1+c_2 \hat{\epsilon}$ of each other.
Thus we can use the guarantees of Lemma~\ref{lem:FindingDecayedMin}
to compute the $\epsilon$-decayed candidates of each bucket.
That is, with high probability we did indeed return the
$\epsilon$-decayed minimum of each bucket $S_{i}^{(t)}$.
The $\epsilon$-decayed minimum
of $k$ sets is the minimum of their respective
$\epsilon$-decayed minimums.
Additionally, Lemma~\ref{lem:trimmingisokay}
shows that trimming our set does not remove
the $\epsilon$-decayed minimum from the set.
So, we have that $u^{(t+1)}$ is the $\epsilon$-decayed minimum over
all the values of $\tilde{d}_{fill}^{( t )}(u )$
with high probability.

Lastly, invoking the bound on distortions incurred by $\epsilon$-decay
from Lemma~\ref{lem:ApproxFactor}, as well as the approximation
error of $\textsc{EstimateDegree}$,
gives that w.h.p.\ the fill degree of $u^{(t+1)}$ is within
$1 + \epsilon$ of the minimum fill degree in $G^{\circ(t)}$.
From all the above high probability claims,
we get a failure probability of at most $n^{-c}$.
So the inductive hypothesis holds for $t+1$ as well.
\end{proof}

We now consider the cost of the algorithms.
For this, we show that if a vertex is close to the $\epsilon$-decayed
minimum, then there is a high chance that it is the $\epsilon$-decayed
minimum. That is to say, if the algorithm queries the
approximate degree of a vertex, there is a good chance that
this vertex belongs to the $\epsilon$-decayed approximate
degree sequence.

\begin{lemma}
\label{lem:CandidateChosen}
For any constant $c_3$, a choice of $\hat{\epsilon}$ (as in
line~\ref{algline:smallerror} of \textsc{ApproxMinDegreeSequence}),
a set of values $d_1,d_2, \dots, d_n$,
and any index $i$, we have
\begin{align*}
& \prob{\delta_1\dots\delta_{n} \sim \hat{\epsilon} \cdot \Exp(1)}
{\text{$i$ is the $\epsilon$-decayed minimum of $d_1, d_2, \dots, d_n$}}\\
  &\hspace{2.5cm} \geq \exp\left( -2 c_3 \right)
\prob{\delta_1 \ldots \delta_{n} \sim \hat{\epsilon} \cdot \Exp(1)}
{
\left( 1 - \delta_{i} \right) d_{i}
\leq \left( 1 +  \frac{\epsilon}{c_1 \log{n}} \right)^{c_{3}}
\min_{j} \left( 1 - \delta_{j} \right) d_{j}
}.
\end{align*}
\end{lemma}

\begin{proof}
Consider generating $\delta_{i}$ last.
Then consider the value
\[
m_{\setminus i}
\defeq
\min_{j \neq i} \left(1 - \delta_{j} \right) d_j.
\]
If $m_{\setminus i} \geq d_i$, then both sides are $1$ and the
result holds trivially.
Otherwise, we condition on
\[
\left(1 - \delta_{i} \right) d_i
\leq \left( 1 +  \hat{\epsilon} \right)^{c_{3}} m_{\setminus i},
\]
or equivalently
\[
\delta_{i}
\geq
\hat{\gamma},
\]
for some $\hat{\gamma}$ such that
\[
\left(1 -  \hat{\gamma} \right) d_i
=
\left( 1 +  \hat{\epsilon} \right)^{c_{3}} m_{\setminus i}.
\]

Then by the memoryless property of exponentials from
Lemma~\ref{lem:OrderStatisticExp}, with probability
at least $\exp(-2 c_3)$, we have
\[
\delta_{i}
\geq \hat{\gamma} + 2 c_3,
\]
which when substituted back in gives
\begin{align*}
\left(1 - \delta_{i} \right) d_i
&\leq
\left( 1 - \hat{\gamma} - 2 c_3 \right) d_{i}\\
&\leq
\left( 1 - 2 c_3 \right)
\left( 1 - \hat{\gamma} \right) d_{i}\\
&=
\left( 1 - 2 c_3 \right)
\left( 1 +  \hat{\epsilon} \right)^{c_{3}} m\\
&\leq m.
\end{align*}
So conditioned on the decayed value of $i$
being within the given threshold of the minimum,
it would decay below the minimum with probability at least $\exp(-2 c_3)$.
\end{proof}

Substituting the value of $c_3 = 7$ as in algorithm \textsc{ApproxMinDegreeSequence}
in Figure~\ref{fig:ApproxMinDegreeSequence},
we get the following corollary.

\begin{corollary}\label{corollary:trimGlobalCandidates}
	If a vertex $v$ is in $\variable{global\_candidates}^{(t)}$
	after line~\ref{algline:trimGlobalCandidates} of \textsc{ApproxMinDegreeSequence},
	then with constant ($\exp(-14)$) probability,
	$v$ is the $\epsilon$-decayed minimum.
\end{corollary}

\noindent
We can now prove our main result.
\begin{proof}[Proof of Theorem~\ref{thm:ApproxMinDegree}.]
The correctness follows from Lemma~\ref{lem:Correctness}.
Theorem~\ref{thm:ApproxDegreeDS} allows us to maintain
access to all the buckets in a total time of
\[
O\left(m \log^{3}n \hat{\epsilon}^{-2} \right)
=
O\left(m \log^{5}n \epsilon^{-2} \right)
\]
across the sequence of pivots.

It remains to bound the costs of the calls to $\textsc{EstimateDegree}$.
By Theorem~\ref{thm:DegreeEstimation}, the total costs of maintaining the
graphs under pivots is $O(m \log^2{n})$, and comes out to be a lower order term.
For the cost of the calls to $\textsc{EstimateDegree}$,
we utilize Corollary~\ref{corollary:trimGlobalCandidates}, which states
that if a vertex is in $\variable{global\_candidates}$,
then it is the one pivoted with constant probability.
Specifically, we prove inductively based on the number of vertices
that remain that the expected cost of calling $\textsc{EstimateDegree}$
is bounded by
\[
c_4 \cdot 
\left( 
\sum_{u \in V^{G^{\left(t\right)}}_{remain} }
d_{remain}^{V^{G^{\left(t\right)}}} \left( u \right)
\right)
\log^{2}{n} \cdot \hat{\epsilon}^{-2},
\]
for some constant $c_4$.

The base case of $t = n$ follows from the lack of vertices remaining.
Suppose the result is true for $t + 1$ vertices.
Then the cost of the next step is bounded by
\begin{multline*}
\sum_{u} \prob{\delta_{v}: v \in V}{\text{$u$ is the $\epsilon$-decayed minimum}} \cdot
c_4 \cdot 
\left( 
\sum_{w \in V^{G^{\left(t\right)}}_{remain} }
d_{remain}^{V^{G^{\left(t\right)}}} \left( w \right)
- d_{remain}^{V^{G^{\left(t\right)}}} \left( u \right)
\right)
\log^{2}{n} \cdot \hat{\epsilon}^{-2}\\
=
c_4 \cdot 
\left( 
\sum_{u \in V^{G^{\left(t\right)}}_{remain} }
d_{remain}^{V^{G^{\left(t\right)}}} \left( u \right)
\right)
-
c_4 \cdot 
\sum_{u} \prob{\delta_{v}: v \in V}{\text{$u$ is the $\epsilon$-decayed minimum}}
\cdot
\left( d_{remain}^{V^{G^{\left(t\right)}}} \left( u \right) \right)
\log^{2}{n} \cdot \hat{\epsilon}^{-2}.
\end{multline*}

\noindent
On the other hand,
we evaluate $\textsc{EstimateDegree}(u, \widehat{\epsilon})$ on $G^{(t)}$
if $u \in \variable{global\_candidates}^{(t)}$.
By Corollary~\ref{corollary:trimGlobalCandidates}, we have
\[
 \prob{\delta_{v}: v \in V}{u \in \variable{global\_candidates}^{(t)}}
\leq
\exp(14)
\cdot
 \prob{\delta_{v}: v \in V}{\text{$u$ is the $\epsilon$-decayed minimum}}.
\]
Therefore, the expected cost of these calls is
\begin{multline*}
\sum_{u}
 \prob{\delta_{v}: v \in V}{u \in \variable{global\_candidates}^{(t)}} \cdot 
c_3 \left( d_{remain}^{V^{G^{\left(t\right)}}} \left( u \right) \right)
\log^{2}{n} \cdot \hat{\epsilon}^{-2}\\
\leq c_3 \cdot \exp(14) \cdot
\sum_{u} \prob{\delta_{v}: v \in V}{\text{$u$ is the $\epsilon$-decayed minimum}}
\cdot
\left( d_{remain}^{V^{G^{\left(t\right)}}} \left( u \right) \right)
\log^{2}{n} \cdot \hat{\epsilon}^{-2},
\end{multline*}
so the inductive hypothesis holds for $t$ as well by letting
\[
c_4 = c_3 \cdot \exp(14).
\]
As the initial total of remaining degrees is $O(m)$,
the total cost of these steps is
\[
O\left( m \log^{2}{n} \hat{\epsilon}^{-2} \right)
=
O\left( m \log^{4}{n} \epsilon^{-2} \right),
\]
which completes the proof.
\end{proof}

\section{Estimating the Fill Degree of a Single Vertex}
\label{sec:DegreeEstimation}

This section discusses routines for approximating
the fill-degree of a single vertex in a partially
eliminated graph.
Additionally, we also need to maintain this
partially eliminated graph throughout the
course of the algorithm.
Specifically, we prove Theorem~\ref{thm:DegreeEstimation}.

Note that in this partially eliminated graph
(which we call the `component graph'),
connected components of the eliminated vertices
are contracted into single vertices, which we now 
call `component' vertices, while the rest of the
vertices are termed `remaining' vertices.
Hence, we can think of the state of the graph as one where the component
vertices form an independent set.
Also, we are only trying to approximate the fill degree of a single
remaining vertex $u$.
The fill degree of $u$ is simply the number of remaining neighbors of $u$
in addition to the number of remaining neighbors of any component
neighbor of $u$.
Since, the former is easy to compute,
the object in question is the cardinality of the
unions of the remaining neighbors of the neighbors of $u$.
This set-of-sets structure also has a natural interpretation
as a matrix.

In particular, if we write out the neighbors of $u$
as rows of some matrix $A$, and view all remaining vertices
as columns of this matrix,
the problem can be viewed as querying for the number
of non-zero columns in a $0/1$ matrix.

Given a $0/1$ matrix $A$ with $r$ rows,
our goal is to estimate the number of non-zero columns,
or columns with at least one entry,
by making the following two types of queries:
\begin{enumerate}
	\item $\textsc{RowSize}(i)$: return the number of non-zero elements in $i$-th row of $A$;
	\item $\textsc{SampleFromRow}(i)$: return a column index $j$
        uniformly sampled among all non-zero entries in row $i$ of $A$.
    \item $\textsc{QueryValue}(i, j)$: returns the value of $A(i, j)$.
\end{enumerate}

The main result as a matrix sampler is:
\begin{lemma}
\label{lem:NonZeroColumnEstimator}
There is a routine $\textsc{EstimateNonZeroColumns}$ that takes (implicit)
access to a matrix $A$, along with access to the three operations
above, $\textsc{RowSize}$, $\textsc{SampleColumn}$, $\textsc{QueryValue}$,
along with an error threshold $\epsilon$,
returns a value that's an $1 \pm \epsilon$ approximation to the number
of non-zero columns in $A$ with high probability.
Furthermore, the expected total number of operations called is
$O(r \log^2{n} \epsilon^{-2} )$
where $r$ is the number of rows of $A$.
\end{lemma}

First, we will prove a weaker version of this result
in Section~\ref{subsec:DegreeEstimation_Matrix}.
This algorithm relies on a routine to
estimate the mean of a distribution,
which is detailed in Section~\ref{subsec:MeanEstimation}.
Finally, by a more careful analysis of
both these algorithms,
we prove the exact claim in Lemma~\ref{lem:NonZeroColumnEstimator}
in Section~\ref{subsec:DegreeEstimation_Better}.

But, before proving this matrix based result,
we first verify that this matrix game can be ported
back to the graph theoretic setting as stated in
Theorem~\ref{thm:DegreeEstimation}.
To do so, we need the following tools for querying
degrees and sampling neighbors in a component graph
as it undergoes pivots.

\begin{lemma}
\label{lem:DegreeEstimationDS}
We can maintain a component graph under
pivoting of vertices in a total time of $O(m \log^2{n})$
so that the operations described in Theorem~\ref{thm:DegreeEstimation}
can be performed in $O(\log{n})$ time each.
This component graph grants oracle access that allows for:
\begin{itemize}
\item querying the state of a vertex,
\item querying the component or remaining degree of a vertex.
\item sampling a random remaining neighbor of either a component
or remaining vertex.
\item sampling a random component vertex.
\end{itemize}
\end{lemma}

We defer the proof of this lemma to
Section~\ref{sec:DynamicGraphs},
along with the corresponding running time guarantees.
Assuming the correctness of Lemmas~\ref{lem:NonZeroColumnEstimator}
and~\ref{lem:DegreeEstimationDS},
matching the above operations with the required
matrix operations described in Theorem~\ref{thm:DegreeEstimation}
then gives its proof.

\begin{proof}(Of Theorem~\ref{thm:DegreeEstimation})

The provided graph theoretic operations can simulate the matrix
operations by.
\begin{enumerate}
\item Generating a list of all the component neighbors of $u$.
\item For each component neighbor $x$, finding the remaining degree of $x$.
\item Finding a random non-zero in some row corresponding
is the same as sampling a remaining neighbor of the component
vertex corresponding to it.
\item To query whether some row/column pair $x$ and $u$ are connected,
we search for $u$ in the list of neighbors for $x$.
Maintaining all neighbors in a searchable data structure such as a binary
search tree resolves this.
\end{enumerate}

Substituting in the runtime bounds gives the desired result.
\end{proof}

\subsection{Column Count Approximator Using Distribution Mean Estimators}
\label{subsec:DegreeEstimation_Matrix}

We start by defining an overall estimator
which is what we eventually sample.
Consider weighting each entry $(i, j)$ by
\[
\frac{1}{\text{ColumnSum}_{A}(j)},
\]
where
\[
\text{ColumnSum}_{A}(j)
\defeq
\sum_{i} A(i, j).
\]

This can be checked to be an unbiased estimator
of the number of non-zero columns.
\begin{lemma}
\label{lem:EstimatorCorrectness}
The number of non-zero columns of a $0/1$-matrix $A$ equals
\[
\sum_{(i, j): A(i, j) = 1}
\frac{1}{\textsc{ColumnSum}_{A}(j)}.
\]
\end{lemma}

\begin{proof}

\[
\sum_{\left(i,j\right)
  :
  A\left(i,j\right) = 1}
{
  \frac{1}{{\textsc{ColumnSum}_A\left(j\right)}}
}
=
\sum_{
\substack{
  j:~\text{column $j$ of $A$ is non-zero}\\
  1 \le i \le r
}}
{
  \frac{A\left(i,j\right)}{{\textsc{ColumnSum}_A\left(j\right)}}
}.
\]
By the definition of column sum,
\[
\sum_{1 \le i \le r} {A(i,j)} = \textsc{ColumnSum}_A(j).
\]

Substituting this back give us
\[
\sum_{\left(i,j\right):A\left(i,j\right) \ne 0}
  {\frac{1}{{\textsc{ColumnSum}_A\left(j\right)}}}  
= 
\left| {
  \left\{ {j:~\text{column $j$ of $A$ is non-zero}} \right\} 
} \right|. \qedhere
\]
\end{proof}

This implies that we only need to estimate column sums.
The way we will actually compute this approximation
is to estimate the mean of some appropriately chosen
distribution.
The pseudocode for such a routine
(\textsc{EstimateMean}) is given in
Figure~\ref{fig:EstimateMean}.

\begin{figure}[H]
	
	\begin{algbox}
		
		$\textsc{EstimateMean}(D, \sigma)$
		
		\underline{Input}: access to a distribution $D$,
		cutoff point $\sigma$.
		
		\underline{Output}: estimate of mean.
		
		\begin{enumerate}
			
			\item Initialize $\textsc{counter} \leftarrow 0$ and $\textsc{sum} \leftarrow 0$;
			\item While {$\variable{sum} < \sigma$}
			\begin{enumerate}
				\item Generate $\variable{x} \sim D$,
				\item $\variable{sum} \leftarrow \variable{sum} + x$,
				\item Increment counter, $\variable{counter} \leftarrow \variable{counter} + 1$.
			\end{enumerate}
			\item Return $\sigma/\variable{counter}$;
		\end{enumerate}
		
	\end{algbox}
	
	\caption{Pseudocode for Mean Estimation}
	
	\label{fig:EstimateMean}
	
\end{figure}

The following lemma bounds the accuracy
and the running time of \textsc{EstimateMean}.
Its proof is detailed in
Section~\ref{subsec:MeanEstimation}.
\begin{lemma}
\label{lem:MeanEstimation}
Let $D$ be any arbitrary distribution on $[0, 1]$
with (unknown) mean $\overline{\mu}$,
and a cut off parameter $\sigma > 0$.
Then for any $\epsilon$, running $\textsc{EstimateMean}(D, \sigma)$, with probability at least $1 - 2\exp ( { - \frac{{{\varepsilon ^2}\sigma }}{5}} )$:
\begin{enumerate}
\item
\label{part:RunTime}
queries $D$ at most
\[
O\left(\frac{\sigma}{ \overline{\mu}}\right)
\]
times;
\item
\label{part:Accuracy}
Produces an output $\widetilde{\mu}$ such that for any
$\epsilon$, we have
\[
\left( 1 - \epsilon \right) \overline{\mu}
\leq \widetilde{\mu}
\leq
\left( 1 + \epsilon \right) \overline{\mu}.
\]
\end{enumerate}
\end{lemma}

An immediate corollary of this is a routine for estimating
the sum of all elements in a column (\textsc{ApproxColumnSum}),
where the runtime depends on the column sum itself.
Its pseudocode is given in Figure~\ref{fig:EstimateColumnSum}.
Lemma~\ref{lem:ColumnSumEstimation} gives the correctness
and running time of \textsc{ApproxColumnSum}.

\begin{figure}[H]
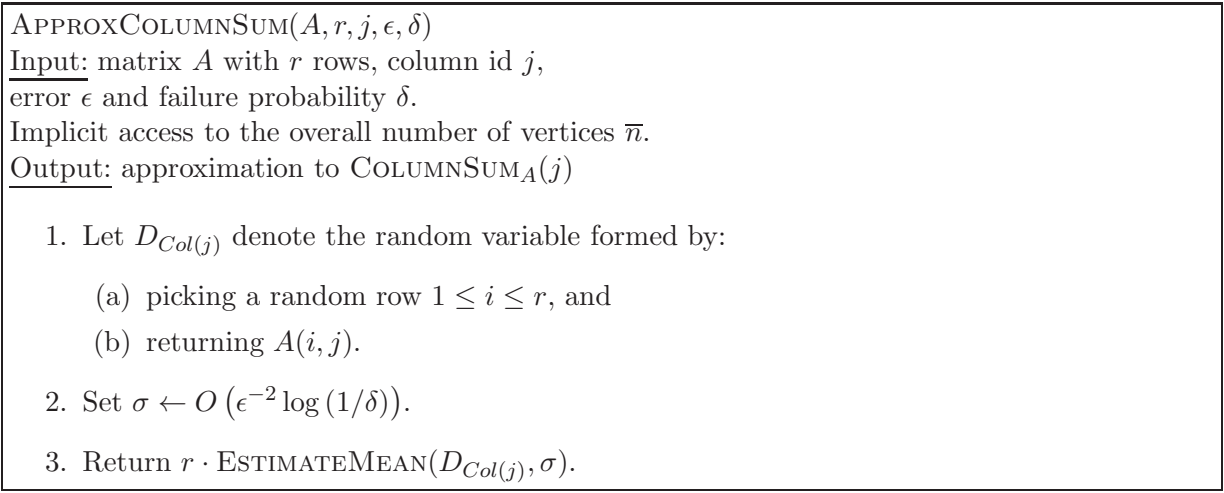

	
	\begin{algbox}
		
		$\textsc{ApproxColumnSum}(A, r, j, \epsilon, \delta)$
		
		\underline{Input:} matrix $A$ with $r$ rows,
		column id $j$,\\
		error $\epsilon$ and failure probability $\delta$.\\
		Implicit access to the overall number of vertices $\overline{n}$.

		\underline{Output:} approximation to $\textsc{ColumnSum}_{A}(j)$
		
		\begin{enumerate}
			
			\item Let $D_{Col(j)}$ denote the random variable formed by:
			\begin{enumerate}
				\item picking a random row $1 \leq i \leq r$, and
				\item returning $A(i, j)$. 
			\end{enumerate}
			\item Set $\sigma \leftarrow O\left( \epsilon^{-2} \log\left( 1 / \delta \right) \right)$.
			\item Return $r \cdot \textsc{EstimateMean}(D_{Col(j)}, \sigma)$.
			
		\end{enumerate}
		
	\end{algbox}
	
	\caption{Pseudocode for Estimating the Column Sum of a Matrix}
	
	\label{fig:EstimateColumnSum}
	
\end{figure}

\begin{lemma}
\label{lem:ColumnSumEstimation}
For any a matrix $A$, any column ID $j$,
any error $\epsilon > 0$,
and any failure probability $\delta$,
a call to $\textsc{ApproxColumnSum}$ returns with probability
at least $1 - \delta$ an $(1 + \epsilon)$
approximation to $\textsc{ColumnSum}_{A}(j)$ while making
\[
O\left( \frac{r \log\left( 1/\delta \right) }{\textsc{ColumnSum}_{A}(j) \epsilon^{2}} \right).
\]
oracle accesses to the matrix $A$ in expectation.
\end{lemma}

\begin{proof}
As in the pseudocode of $\textsc{ApproxColumnSum}$
in Figure~\ref{fig:EstimateColumnSum}, we define
the random variable $D_{Col(j)}$ as:
\begin{enumerate}
\item picking a random row $1 \leq i \leq r$, and
\item returning $A(i, j)$. 
\end{enumerate}
This gives a Bernoulli distribution that is:
\begin{itemize}
\item $1$ with probability $\textsc{ColumnSum}_A(j)/r$, and
\item $0$ otherwise.
\end{itemize}
The mean of $D_{\textsc{Col}(j)}$ is its probability of being $1$:
\[
\frac{\textsc{ColumnSum}_A(j)}{r}.
\]
The cost of the call to $\textsc{MeanEstimation}$
is then given by Lemma~\ref{lem:MeanEstimation}.
It gives that the number of accesses to the matrix $A$ 
via $\textsc{QueryValue}$ in Algorithm~\ref{fig:EstimateColumnSum}
is upper-bounded by
\[
O\left( {
  \frac{\sigma } {\left(\textsc{ColumnSum}_A(j)/r\right) }} 
\right) 
= 
O \left( {
  \frac{r \log(1/\delta) } {\textsc{ColumnSum}_A(j) {\varepsilon ^{2}}}
} \right). \qedhere
\]
\end{proof}

This performance means that we can just treat
$(1/\text{ColumnSum}_{A}(j))$ as a random variable,
and sample enough entries so that the sum is approximately
$O(\log{n} \epsilon^{-2})$.
The running time as given in Lemma~\ref{lem:ColumnSumEstimation}
adds an extra factor of $n$ to this, giving the claimed 
running time in Lemma~\ref{lem:NonZeroColumnEstimator}.
Pseudocode of the overall algorithm is given in
Figure~\ref{fig:SlowerEstimateNonZeroColumns}.

\begin{figure}[H]
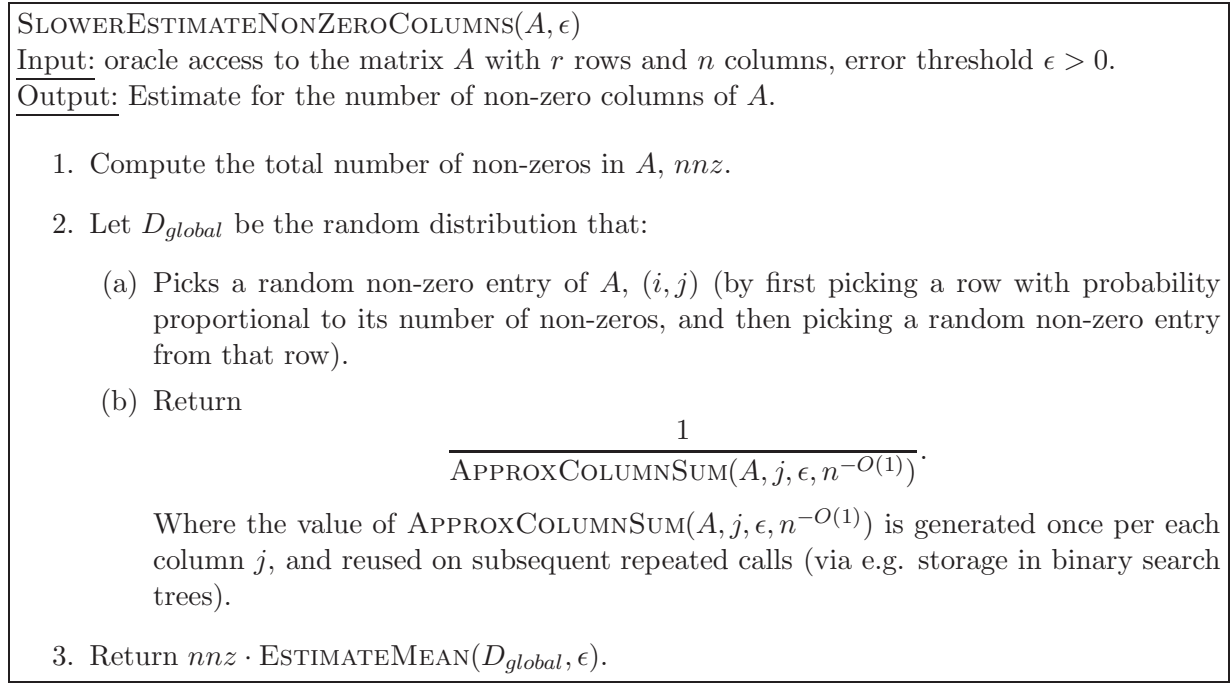

	
	\begin{algbox}
		$\textsc{SlowerEstimateNonZeroColumns}(A, \epsilon)$
		
		\underline{Input:} oracle access to the matrix $A$
		with $r$ rows and $n$ columns,
		error threshold $\epsilon > 0$.
		
		\underline{Output:} Estimate for the number of non-zero columns of $A$.
		
		\begin{enumerate}
			
			\item Compute the total number of non-zeros in $A$, $nnz$.
			\item Let $D_{global}$ be the random distribution that:
			\begin{enumerate}
				\item Picks a random non-zero entry of $A$, $(i, j)$
				(by first picking a row with probability proportional to its
				number of non-zeros, and then picking a random non-zero entry from that row).
				\item Return
				\[
				\frac{1}{\textsc{ApproxColumnSum}(A, j, \epsilon, n^{-O\left( 1 \right)})}.
				\]
				Where the value of $\textsc{ApproxColumnSum}(A, j, \epsilon, n^{-O\left( 1 \right)})$
				is generated once per each column $j$,
				and reused on subsequent repeated calls
				(via e.g. storage in binary search trees).
			\end{enumerate}
			\item Return $nnz \cdot \textsc{EstimateMean}(D_{global}, \epsilon)$.
		\end{enumerate}
		
	\end{algbox}
	
	\caption{Pseudocode for Estimating the Number of Non-Zero Columns of a Matrix}
	
	\label{fig:SlowerEstimateNonZeroColumns}
	
\end{figure}

We first bound the correctness of the result returned by
$\textsc{SlowerEstimateNonZeroColumns}(A, \epsilon)$, and the
expected number times it samples $D_{global(j)}$.

\begin{lemma}
\label{lem:ApproxnonZeroColumnsHelper}
With high probability, the estimate returned by
$\textsc{SlowerEstimateNonZeroColumns}(A, \epsilon)$
is within $1 \pm \epsilon$ of the number of non-zero columns.
\end{lemma}

\begin{proof}
To begin with, we explicitly extract out all the randomness
in Algorithm~\ref{fig:SlowerEstimateNonZeroColumns} considering
running all calls to 
$\textsc{ApproxColumnSum}(A, r, j, \epsilon, n^{-O(1)})$ 
beforehand.

By Lemma~\ref{lem:ColumnSumEstimation}, we have that with
high probability we have for each $j$,
\[
\left( 1 - \epsilon \right)
\textsc{ColumnSum}_A\left( j \right)
\leq
\textsc{ApproxColumnSum}\left(A, r, j, \epsilon, n^{-O\left(1\right)}\right)
\leq
\left( 1 + \epsilon \right)
\textsc{ColumnSum}_A\left( j \right).
\]

So by Lemma~\ref{lem:EstimatorCorrectness}, we have that
the expectation of $D_{global}$, $\mu(D_{global})$, is
within $1 \pm \epsilon$ of the true value with high probability.
Formally,
\[
\left( 1 - \epsilon \right)
\frac{\left|\left\{c: A_{:, c} \neq 0\right\}\right|}{nnz}
\leq
\mu \left( D_{global} \right)
\leq
\left( 1 + \epsilon \right)
\frac{\left|\left\{c: A_{:, c} \neq 0\right\}\right|}{nnz}.
\]
Incorporating the guarantee of Lemma~\ref{lem:MeanEstimation},
part~\ref{part:Accuracy} then gives:
\[
\left( 1 - 3\epsilon \right)
\frac{\left|\left\{c: A_{:, c} \neq 0\right\}\right|}{nnz}
\leq
\textsc{EstimateMean}\left(
  \mu \left( D_{global} \right),
  \epsilon
\right)
\leq
\left( 1 + 3\epsilon \right)
\frac{\left|\left\{c: A_{:, c} \neq 0\right\}\right|}{nnz}.
\]

The desired bound then follows from halving $\epsilon$,
and the final multiplication by $r$ on the last line.
\end{proof}

\begin{proof}(Of Lemma~\ref{lem:NonZeroColumnEstimator}
with a worse factor of $O(\log^{2}n \epsilon^{-4})$)

The correctness follows from Lemma~\ref{lem:ApproxnonZeroColumnsHelper},
so it suffices to bound the total number of queries made to entries of $A$.
Furthermore, Part~\ref{part:RunTime} of Lemma~\ref{lem:MeanEstimation}
gives that the expected number of queries made to
$D_{global}$ is:
\[
O\left(
\frac{nnz \cdot \log{n}}
{\left|\left\{c: A_{:, c} \neq 0\right\}\right| \epsilon^{2}}
\right).
\]
So it suffices to bound the expected cost of each evaluation of $D_{global}$.

Applying Lemma~\ref{lem:ColumnSumEstimation} to every column $j$ gives that
w.h.p. the number of queries made by $\textsc{ColumnSum}_{A}(j)$ is at most
\[
O\left( \frac{r \log{n}}{\textsc{ColumnSum}_{A}\left(j\right) \epsilon^{2}} \right).
\]
Summing this over all $\textsc{ColumnSum}_{A}(j)$ entries in that column,
as well as all the non-zero columns gives that the expected number of queries
when querying for a single entry of $D_{global}$ is:
\begin{multline*}
\frac{1}{nnz}
\sum_{j: \textsc{ColumnSum}_{A}\left(j\right) \neq 0}
\sum_{i: A_{i, j} \neq 0}
O\left( \frac{r \log{n}}{\textsc{ColumnSum}_{A}\left(j\right) \epsilon^{2}} \right)\\
=
\frac{1}{nnz}
\sum_{j: \textsc{ColumnSum}_{A}\left(j\right) \neq 0}
O\left( \frac{r \log{n}}{\epsilon^{2}} \right)\\
=
O\left(
\frac{\left|\left\{c: A_{:, c} \neq 0\right\}\right| \cdot r \log{n}}
{nnz \cdot \epsilon^{2}}
\right).
\end{multline*}
Multiplying this with the expected number of queries to $D_{global}$
then gives the overall result.
\end{proof}
We remark that the runtime bound also holds with
high probability instead of in expectation if we invoke Chernoff bounds.
This is because the cost of each query to $D_{global}$ is bounded
by $O(r \log{n} \epsilon^{-2})$, and the total cost bound is larger
by a factor of at least $\log{n}$.

\subsection{Estimating Mean of a Distribution}
\label{subsec:MeanEstimation}

We now provide the details of the mean estimation algorithm,
which also gives the correctness of the column sum estimation scheme.

We analyze an equivalent scheme which
generates the same output:
\begin{enumerate}
	\item Generate a stream of infinite i.i.d. samples from $D$,
        denoted as $X_1, X_2, \cdots$;
	\item Let $\variable{counter}$ be
        $\mathop {\arg \min }\limits_{t > 0}
        \left\{ {\sum\limits_{i \le t} {{X_i}}  \ge \sigma } \right\}$;
	\item Output $\sigma/\variable{counter}$.
\end{enumerate}
This process evaluates more samples than \textsc{EstimateMean}
(from Figure~\ref{fig:EstimateMean}).
However, the extra evaluations happen after the termination
of that process.
So it does not affect the outcome.

We will bound the success probability by bounding
the partial sum of $\{X_{i}\}$ at two points.
These two points are defined based on the
(hidden) value of $\mu$,
the expectation of the distribution $D$.
For a distribution $D$,
and an error $\epsilon > 0$,
we make two marks at:
\begin{equation}
L_{D, \epsilon}
\defeq
\frac{\sigma }{{(1 + \varepsilon ) \mu }},
\label{eq:L}
\end{equation}
and
\begin{equation}
R_{D, \epsilon}
\defeq
\frac{\sigma }{{(1 - \varepsilon ) \mu }}.
\label{eq:R}
\end{equation}

By some algebra,
we can check that if we terminate with
\[
L_{D, \epsilon} \leq \variable{counter} \leq R_{D, \epsilon},
\]
then the final outcome is good.
So it suffices to bound the probability of $\variable{counter}<L$
and $\variable{counter}>R$ separately.


\begin{lemma}
\label{lem:ProbCounterTooSmall}
For any sequence $X_1, X_2 \ldots $
generated by taking i.i.d. copies of
a random variable $D$,
and with we have
\[
\prob{X_1, X_2, \ldots}
{ {\sum\limits_{1 \le i \le L_{d, \epsilon}} {{X_i}}  \ge \sigma } }
\leq
\exp \left( - \frac{\epsilon^2 \sigma}{4} \right).
\]
\end{lemma}

\begin{proof}
Linearity of expectation gives:
\[
\expec{X_1, X_2 \ldots }
{\sum_{1 \le i \le L_{D, \epsilon}} {{X_i}} }
= \mu L_{D, \epsilon}.
\]
So as $X_{1} \ldots X_{L_{D, \epsilon}}$ are i.i.d.,
we get
\[
\prob{
X_{1}, X_{2} \ldots
}
{
  \sum_{1 \le i \le L} {{X_i}}
  \geq
  \left(1 + \varepsilon \right) \mu L_{D, \epsilon}
}
\leq
\exp \left( - \frac{{{\varepsilon ^2}{\mu} L_{D, \epsilon} }}{3} \right)
\]

which directly imples the lemma 
by taking $(1+\epsilon)\mu L_{D, \epsilon} = \sigma$ into the left-hand side and $\frac{\mu L_{D, \epsilon}}{3} > \frac{\sigma}{4}$ for small enough $\epsilon$ into the right-hand side.

\end{proof}

\begin{lemma}
\label{lem:ProbCounterTooBig}
For any sequence $X_1, X_2 \ldots $
generated by taking i.i.d. copies of
a random variable $D$,
and with we have
\[
\prob{X_1, X_2, \ldots}
{ {\sum\limits_{1 \le i \le R_{d, \epsilon}} {{X_i}}  \leq \sigma } }
\leq
\exp \left( - \frac{\epsilon^2 \sigma}{4} \right).
\]
\end{lemma}

\begin{proof}
Similiar to proof of Lemma~\ref{lem:ProbCounterTooSmall},
but with lower end of Chernoff bounds.
\end{proof}

\begin{proof}(Of Lemma~\ref{lem:MeanEstimation} with an additional overhead of $\epsilon^{-2}$)

Consider the routine $\textsc{SlowerEstimateNonZeroColumns}$
whose pseudocode is in Figure~\ref{fig:SlowerEstimateNonZeroColumns}.
The running time (Part~\ref{part:RunTime})
is an immediate consequence of
the bound on $\variable{counter} \leq R_{D, \epsilon}$
from Lemma~\ref{lem:ProbCounterTooBig}.
So it remains to bound the (Part~\ref{part:Accuracy}).
Recall that the estimator is
\[
\frac{\sigma }{\variable{counter}}
\]
Apply union bound over Lemma~\ref{lem:ProbCounterTooSmall}
and Lemma~\ref{lem:ProbCounterTooSmall}, we get
\[
\prob{X_1, X_2 \ldots}
{
  L_{D, \epsilon} \leq \textsc{Counter} \leq R_{D, \epsilon}
}
\geq
1 - 2 \exp \left(
  -\frac{\epsilon^2 \sigma}{4}
\right)
\geq
1 - \exp \left(
  -\frac{\epsilon^2 \sigma}{5}
\right).
\]

Putting in the definition of $L_{D, \epsilon}$
and $R_{\epsilon}$ gives this is the same as
\[
\frac{\sigma}
{\left( 1 + \epsilon\right) \overline{\mu} }
\leq
\textsc{Counter}
\leq
\frac{\sigma}
{\left( 1 - \epsilon\right) \overline{\mu} },
\]
which is in turn equivalent to
\[
\left( 1 - \epsilon \right) \overline{\mu}
\leq
\frac{\sigma}{\textsc{Counter}}
\leq
\left( 1 + \epsilon \right) \overline{\mu},
\]
or the estimator $\frac{\sigma}{\textsc{Counter}}$
is a good approximation to $\overline{\mu}$.

\end{proof}

\subsection{More holistic analysis with better bounds}
\label{subsec:DegreeEstimation_Better}

We now give a better running time bound by combining
the analyses of the two estimators in a more global,
holistic analysis.
Pseudocode of the final routine is in Figure~\ref{fig:EstimateNonZeroColumns}.

\begin{figure}[H]
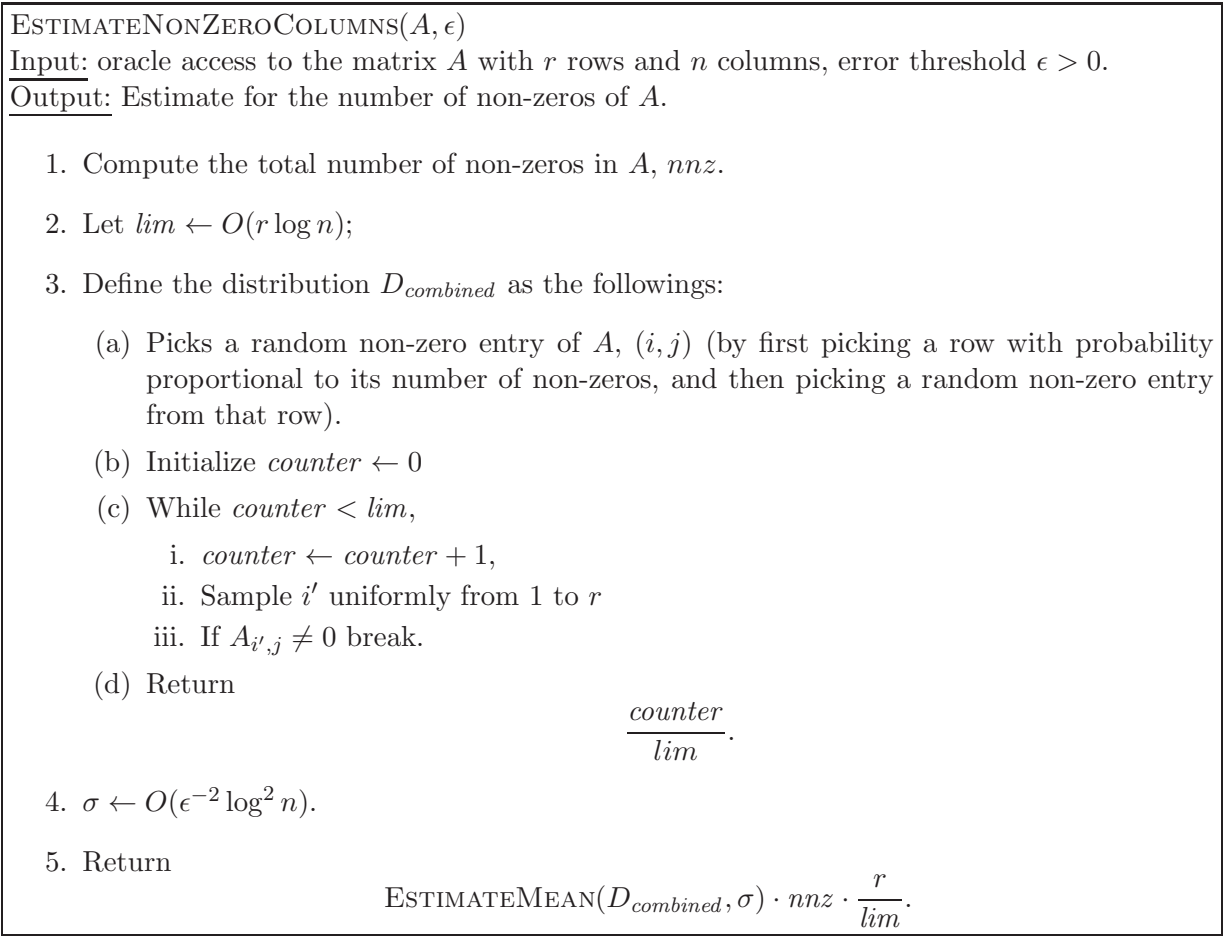


\begin{algbox}
$\textsc{EstimateNonZeroColumns}(A, \epsilon)$

\underline{Input:} oracle access to the matrix $A$
with $r$ rows and $n$ columns,
error threshold $\epsilon > 0$.

\underline{Output:} Estimate for the number of non-zeros of $A$.

\begin{enumerate}

\item Compute the total number of non-zeros in $A$, $nnz$.
\item Let $\variable{lim} \leftarrow O(r \log n)$;
\item Define the distribution $D_{\variable{combined}}$ as the followings:
  \begin{enumerate}
  \item Picks a random non-zero entry of $A$, $(i, j)$
  (by first picking a row with probability proportional to its
  number of non-zeros, and then picking a random non-zero entry from that row).
  \item Initialize $\variable{counter} \leftarrow 0$
  \item While $\variable{counter} < \variable{lim}$,
     \begin{enumerate}
       \item $\variable{counter} \leftarrow \variable{counter} + 1$,
       \item Sample $i'$ uniformly from $1$ to $r$
       \item If $A_{i',j} \neq 0$ break.
	\end{enumerate}
  \item Return
    \[
      \frac{\variable{counter}}{lim}.
    \]
\end{enumerate}
\item $\sigma \leftarrow O(\epsilon^{-2}\log^2n)$.
\item Return 
\[
\textsc{EstimateMean}(D_{\variable{combined}}, \sigma)
\cdot \variable{nnz} \cdot \frac{r}{\variable{lim}}.
\]
\end{enumerate}

\end{algbox}

\caption{Pseudocode for Fast Estimating the Number of Non-Zero Columns of a Matrix}

\label{fig:EstimateNonZeroColumns}

\end{figure}

At the core of this algorithm is the following
simpler, combined distribution.
\begin{definition}
\label{def:DCombined}
We define the combined distribution, $D_{\variable{combined}}$
as the distribution given by:
\begin{enumerate}
\item Sampling a non-zero entry $(i, j)$ from $A$ uniformly at random.
\item Return $1/r$ times the minimum of $O(r \log{n})$ or
the number of random $i'$s picked until a $A_{i', j}$ is non-zero.
\end{enumerate}
\end{definition}

This combined distribution is artifically capped at $1$.
More importantly, we can precisely calculate its expected value,
up to a $1 / poly(n)$ perturbation due to the truncation
at $\variable{lim}$.
\begin{lemma}
\label{lem:DCombinedExpectation}
The distribution $D_{\variable{combined}}$
as defined in Definition~\ref{def:DCombined} has expectation
\[
\left( 1 - \frac{1}{n} \right)
\frac{r \cdot \left|\left\{c: A_{:, c} \neq 0\right\}\right|}{nnz \cdot \variable{lim}}
\leq
\expec{}{D_{\variable{simple}}} 
\leq
\frac{r \cdot \left|\left\{c: A_{:, c} \neq 0\right\}\right|}{nnz \cdot \variable{lim}},
\]
and the expected cost of sampling $D_{\variable{simple}}$ is
\[
O\left( \frac{r \cdot \left|\left\{c: A_{:, c} \neq 0\right\}\right|}{nnz} \right).
\]
\end{lemma}

\begin{proof}

For each column $j$, denote $n_j$ to be the number of 
non-zeros in column $j$ of $A$. 
For the ease of representation,
we use $p$ to denote the probability of picking a
non-zero from this column,
\[
p \defeq \frac{n_j}{r}.
\]
As there is at least one non-zero entry 
on the column $j$ in step 3,
we should assume $n_j \geq 1$ for all time.
Hence, $p$ is always positive.

Next, define the event $\textsc{Hit}_j$ 
as getting a non-zero $A_{i', j}$ by uniformly sampling $i'$ at 3(a).
Let $h_j$ be the number of independent repeats of $\textsc{Hit}_j$ 
to make one happening without restriction of the iterations.
Then we have, for any integer $k$,
\[
\prob{h_j}{h_j=k}
= {\left( {1 - p} \right)^{k - 1}} \cdot p
\]
So its expectation is given by
\[
\expec{h_{j}}{h_j}
= 
\sum_{1 \le k}
\prob{h_j}{h_j = k} \cdot k
= 
\sum_{1 \le k}
\left(1 - p\right)^{k - 1} \cdot p \cdot k
\]

To compute this value, consider the generating function
\[
G\left(x\right)
\defeq
\sum_{1 \le k} x^{k - 1} \cdot k.
\]
Its integral is:
\[
\int_0^x {G\left(y\right)dy}  
= 
\int_0^x {\left( {\sum_{1 \le k} {{y^{k - 1}}k} } \right) \cdot dy}  
= 
\sum\limits_{1 \le k} {\int_0^x {{y^{k - 1}}kdy} }  
= 
\sum\limits_{1 \le k} {{x^k}}  
= 
\frac{x}{{1 - x}}.
\]

Then we have
\[
G\left(x\right) 
= 
\frac{{d\left( {\int_0^x {G(y)dy} } \right)}}{{dx}} 
= 
\frac{{d\left( \frac{x}{1 - x} \right)}}{{dx}} 
= 
\frac{1}{{{{(1 - x)}^2}}} .
\]
Taking it back to the expectation, we get
\[
\expec{}{h_j} 
= p \cdot G(1 - p)
= p \cdot \frac{1}{p^2}
= \frac{1}{p}.
\]
This means if our initial entry is in column $j$,
the expected value of 

To account for the truncation, note that
\[
\prob{h_j}{{h_j} \ge m}
= 
{\left( {1 - p} \right)^m} 
= 
O\left( {{n^{ - O\left(1\right)}}} \right).
\]
So this changes the expectation by at most $n^{-O(1)}$.
So we have
\[
\frac{r}{n_j \cdot \variable{lim}}
- n^{-O\left(1\right)}
\leq 
\expec{D_{\variable{combined}}
  \left| \text{column $j$ is picked} \right.
}{D_{\variable{simple}}}
\leq
\frac{r}{n_j \cdot \variable{lim}}
\]
and substituting back the probability that we
pick column $j$ with probability
\[
\frac{n_j}{nnz}
\]
gives
\[
\frac{r \cdot \left|\left\{c: A_{:, c} \neq 0\right\}\right|}{nnz \cdot \variable{lim}} - n^{-O\left(1\right)}
\leq
\expec{}{D_{\variable{combined}}} 
\leq
\frac{r \cdot \left|\left\{c: A_{:, c} \neq 0\right\}\right|}{nnz \cdot \variable{lim}},
\]
and accuracy bound follows from the observation that the numerator is at least $1$
and the denominator is bounded by $nnz \cdot \lim \leq n^{4}$.

The expected running time also follows similarly,
except we do not divide the number of terms sampled by $\variable{lim}$.
\end{proof}

This means we can then invoke $\textsc{EstimateMean}$
to approximate the mean of this distribution,
and thus gives our overall guarantees.

\begin{proof}(Of Lemma~\ref{lem:NonZeroColumnEstimator} )

Notice that the value of $D_{\variable{simple}}$
is always between $0$ and $1$ due to the truncation by $\variable{lim}$,
and then dividing by it.
By Lemma~\ref{lem:MeanEstimation}, Part~\ref{part:Accuracy},
we have that with high probability we obtain a
 $1+\epsilon$ approximation of its expectation,
which after the multiplication by $nnz \frac{\variable{lim}}{r}$
gives an $1 + \epsilon$ approximation to the number
of non-zero columns.

It remains to bound the running time costs.
Lemma~\ref{lem:MeanEstimation}, Part~\ref{part:RunTime}
gives that the expected number of times we sample
$D_{\variable{combined}}$ is at most
\[
\frac{nnz \cdot \variable{lim} \cdot \log{n} \epsilon^{-2}}
{r \cdot \left|\left\{c: A_{:, c} \neq 0\right\}\right|},
\]
while the expected cost per call is
\[
O\left( \frac{r \cdot \left|\left\{c: A_{:, c} \neq 0\right\}\right|}{nnz} \right).
\]
Multiplying these gives that the expected total cost is
\[
O\left( \variable{lim} \log{n} \epsilon^{-2} \right)
= O\left(r \log^2{n} \epsilon^{-2} \right).
\]
Furthermore, since the cost per call is capped at $\lim = O(r \log{n})$,
we also get that the runtime cost is concentrated around this value
with high probability.
\end{proof}

\section{Maintaining Graphs Under Pivots}
\label{sec:DynamicGraphs}

In this section we show that both the random graph access
operations described in Theorem~\ref{thm:DegreeEstimation},
and the $\ell_0$-estimator as described in Definition~\ref{def:Sketch}
can be maintained efficiently under pivoting operations.

We start by checking that the component graph can be
maintained under pivoting of new vertices while providing
random-access to remaining and component neighbors of any vertex.

\begin{proof} (of Lemma~\ref{lem:DegreeEstimationDS})
We will maintain the adjacency list of $G^{\circ}$
explicitly, with each node storing its state as well
as its neighbors in balanced binary search trees.

When we pivot a vertex, we examine all its neighbors
that are component vertices, and merge the neighborhood
lists of these vertices.
By always inserting elements from the smaller list
into the larger, we can guarantee that each element
is inserted at most $O(\log{n})$ times.
So across all $m$ edges the total cost is $O(m \log^2{n})$.

When a vertex is pivoted, we also need to move it
from the remaining vertex list to the component vertex
list in all of its neighbors.
This can be done by going through all the edges
of the vertex once.
Its cost is $O(m)$ because each vertex is pivoted
at most once, and prior to pivoting no edges are added
to it.

These binary search trees with all the neighbors allow
for the sampling of a random component/remaining
neighbors in $O(\log{n})$ time.
A global list maintaining all the component/remaining
vertices also allow them to be sampled in $O(\log{n})$ time.
\end{proof}

The $\ell_0$-sketches are maintained similarly
as the graph changes.
However, the minimum sketch value in each neighborhood
keeps changing, and we deal with this by propagating
new values proactively across to neighbors.

The algorithm is based on the notion of an eager-propagation
routine: every time the $x_{\min}$ at some vertex changes,
it informs its neighbors of this change, and then
in turn propagates this value.

This routine in the static case is the same as the $\ell_0$
estimators.
Our main modification is to make it dynamic:
after each pivot, the minimum per vertex can increase.
Figure~\ref{fig:DynamicGlobalVar} contains a brief description of data structures we use to achieve this.

\begin{figure}[H]
  \begin{algbox}
    \textsc{Maintain graphs under pivots}
    
    \underline{Additional Variables}:
    graph $G$ that's undergoing pivots,
    .
    
    \begin{enumerate}
      \item Set $V_{remaining}$ containing the remaining vertices.
      \item Set $V_{component}$ containing the component vertices.
      
      \item For each component vertex $z$, an associated min-heap
      \[
      z.\mathit{remaining}
      \] that contains the keys of its remaining neighbors.

  	  \item For each remaining vertex $u$, a min-heap
      \[
         u.\mathit{fill}
      \]
      that contains the union of $z.\mathit{remaining}.\mathit{min}$
      for each component vertex $z$ adjacent to $u$
      as well as the keys of $u$'s remaining neighbors.
    \end{enumerate}
    
  \end{algbox}
  
  \caption{Global Variables for maintaining data structures under vertex pivots}
  \label{fig:DynamicGlobalVar}
\end{figure}

In the case where nothing gets deleted,
the following observation is sufficient for bounding the
cost of the propagations.
\begin{lemma}
\label{lem:MinChanges}
For any sequence of increasing sets
\[
S^{(0)}, S^{(1)}, \ldots S^{(t)},
\]
the expected number of different minimums over a
random labeling of the elements to $[0, 1]$
is $O(\log{n})$.
\end{lemma}

The major difficulty dealing with this is that
deletions reduce degrees.
In particular, it is possible for the min at some vertex
to change $\Omega(n)$ times due to repeated deletions.

As a result, we can only bound the total, or average
number of propagations.
This leads to a much more involved amortized
analysis, where we also use backwards analysis to explicitly
bound the probability of each informing operation.

Given a component graph $G^{(t)}$ and a (remaining)
vertex $u$ to be pivoted, we use the routine \textsc{PivotVertex}
to produce a new graph $G^{(t+1)}$.
In terms of the structure of the graph, our routine does the same
thing as the traditional quotient graph model for symmetric
factorization~\cite{GeorgeL81}.

Therefore we turn our attention to the problem of maintaining the minimum sketch
values of the neighborhoods, specifically the values
$\xmin(N^{G^{(t)}}_{remaining}(v))$ and
$\xmin(N^{G^{(t)}}_{fill}(v))$.
This update procedure is basically a notification mechanism.
When the status of a vertex changes, we update the data structures of its
neighbors correspondingly.
The \textit{Fill} heap will then give $\xmin(N_{fill}(u))$ and be
used to estimate the fill-degree of each remaining vertex
as described in Section~\ref{sec:Sketching}.

These two sets of heaps are then maintained via a notification mechanism.
Suppose a remaining vertex $v$ is pivoted.
Then, for a component vertex $w$, the content of $w.\mathit{remaining}$
changes only if $v$ is its neighbor.
Pseudocode of this update is given in Figure~\ref{fig:PivotVertex}.
In particular, since $v$ is no longer a remaining vertex, its entry needs to be removed from
$z.\mathit{remaining}$.
Furthermore, if $x_v$ was the minimum element in $w.\mathit{remaining}$,
this is no longer the case and the other remaining neighbors of $w$ need to be notified of
this (so they can update their $\mathit{fill}$ heaps).
This is done via the call to $\textsc{Informremaining}$ in
Line~\ref{algline:pivot-inform-remaining} of the algorithm.
The last step consists of melding the (now component) vertex $v$ with its existing component
neighbors via calls to $\textsc{Meld}$.

\begin{figure}

\begin{algbox}
$\textsc{PivotVertex}(u)$

  \underline{Input}: (implicitly as a global variable)
  a graph state $G = \langle {V_{Remaining}},{V_{Component}},E\rangle$
  along with associated data structures.\\
  A vertex $v \in {V_{Remaining}}$ to be pivoted, 

  \underline{Output}: A list of vertices whose values of $x_{\min} (N_{fill}(u))$ have changed.

\begin{enumerate}
\item Initialize $\mathit{changed\_list} \leftarrow \emptyset$.
\item For each vertex $w \in N_{component}^G(v)$ in lexicographical order
\begin{enumerate}
\item  $w.\mathit{remaining}.\textsc{Remove}(x_v)$.
        \label{algline:remove-from-component}
\item If $x_v$ was the old minimum in $w.\mathit{remaining}$
\begin{enumerate}
      \item $\mathit{changed\_list} \leftarrow \mathit{changed\_list}
          \cup \textsc{Informremaining}(w,x_v,w.\mathit{remaining}.\mathit{top})$.
        \label{algline:pivot-inform-remaining}
\end{enumerate}
\item  $v.\mathit{remaining} \gets \textsc{Meld}(v, w)$;
\end{enumerate}
\item Update edges and $V_{Component}$ and $V_{Remaining}$ to form $G'$;
\item Return $\mathit{changed\_list}$.
\end{enumerate}
\end{algbox}

\caption{Pseudocode for pivoting a vertex}
\label{fig:PivotVertex}

\end{figure}

\begin{figure}

\begin{algbox}
$\textsc{Informremaining}(w, x_{\mit{old}}, x_{\mit{new}})$

  \underline{Input}: (implicitly as a global variable)
  a graph state $G = \langle {V_{Remaining}},{V_{Component}},E\rangle$
  along with associated data structures.\\
  A `source' component vertex $w \in V_{Component}$ that's causing updates,\\
  old and new values $x_{\mit{old}}$ and $x_{\mit{new}}$.
  
  \underline{Output}: A list of vertices whose $v.\mathit{fill}.\min$ changed.

\begin{enumerate}
\item Initialize $\mathit{changed\_list} \leftarrow \emptyset$.
\item For each $v \in N^{G}_{remaining}(w)$
\begin{enumerate}
      \item Remove the entry with key $x_{\mit{old}}$ associated to $w$
      from $v.\mathit{fill}$;
      \item Add an entry with key $x_{\mit{new}}$ associated to $w$ to
      $v.\mathit{fill}$;
      \item If $v.\mathit{fill}.\mathit{min}$ changed,
      $\mathit{changed\_list} \leftarrow \mathit{changed\_list} \cup \{ v \}$.
\end{enumerate}
\item Return $\mathit{changed\_list}$.
\end{enumerate}

\end{algbox}

  \caption{Pseudocode for propagating to remaining vertex neighbors}
  \label{fig:InformRemaining}
\end{figure}

The routine \textsc{Informremaining} (Algorithm~\ref{fig:InformRemaining}) is
responsible for updating the contents in the \textit{Fill} heaps of remaining
vertices.
We break down its cost into two parts: when it is invoked by
\textsc{PivotVertex}, and when it is invoked by \textsc{Meld}.
The first type of calls happens only when a remaining vertex $v$ is pivoted, and $v$
is the minimum entry of the \textit{remaining} heap of a component vertex.
The following lemma gives an upper bound to the expected cost of such calls to
\textsc{Informremaining} by arguing that this event happens with low
probability.

\begin{lemma}
\label{lem:RemainingUpdatesDueToDeletion}
  The expected total number of updates to remaining vertices made by
  \textsc{Informremaining} when invoked from \textsc{PivotVertex}
  (Line~\ref{algline:pivot-inform-remaining})
  over any sequence of $n$ pivots that are independent of the values of the
  $x_u$s is $O(m)$.
\end{lemma}

\begin{proof}
  Let $G$ be any state during the sequence, and let $v$ be the remaining vertex to
  be pivoted with $w$ as a neighboring component vertex.
  We only invoke \textsc{Informremaining} if $x_v$ is the minimum value in
  $w.\mathit{remaining}$, which happens with probability
  $1/|N^{G}_{remaining}(w)|$ and would cost
  $O(|N^{G}_{remaining}(w)|\log n)$.
  Therefore the expected cost is only $O(\log n)$ for a pair of remaining vertex
  $v$ and neighboring component vertex $w$.
  When a remaining vertex $v$ is pivoted, its degree is the same as in the
  original graph.
  Therefore the number of such $v,w$ pairs is bounded by the degree of $v$
  and hence the total expected cost is $O \left( \sum_{v \in V} \deg(v) \log n \right) = O(m\log n)$.
\end{proof}

The calls to $\textsc{Meld}$ are the primary bottlenecks in the running time,
but will be handled similarly.
Its pseudocode is given in Figure~\ref{fig:Meld}.

\begin{figure}

\begin{algbox}
$\textsc{Meld}(u, v)$

\underline{Input}: (implicitly as a global variable)
A graph state $G = \langle {V_{Remaining}},{V_{Component}},E\rangle$
along with associated data structures.\\
Two component vertices $u$ and $v$ to be melded. 
  
\underline{Output}: None. The algorithm simply updates the global state.

\begin{enumerate}
\item If
$u.\mathit{remaining}.\mathit{min} < v.\mathit{remaining}.\mathit{min}$
\begin{enumerate}
  \item $\textsc{Informremaining}(v,v.\mathit{remaining}.\mathit{min},u.\mathit{remaining}.\mathit{min})$;
      \label{algline:informvaboutu}
\end{enumerate}
\item Else If
$v.\mathit{remaining}.\mathit{min} < u.\mathit{remaining}.\mathit{min}$
\begin{enumerate}
  \item $\textsc{Informremaining}(u,u.\mathit{remaining}.\mathit{min},v.\mathit{remaining}.\mathit{min})$;
  \label{algline:informuaboutv}
\end{enumerate}
\item $\textsc{HeapMeld}(v.\mathit{remaining},u.\mathit{remaining})$;
\end{enumerate}

\end{algbox}

\caption{Pseudocode for melding two component vertices,
and informing their neighbors of any changes in the minimizers
of $N_{remaining}$.}
\label{fig:Meld}

\end{figure}

We will show that the expected number of vertices updated by
$\textsc{Informremaining}$ that result from any fixed sequence of calls to
$\textsc{Meld}$ is bounded by $O(m \log{n})$.
We first analyze the number of updates during a single meld in the following
lemma.

\begin{lemma}
  \label{lem:remaining-updates}
  Let $u$ and $v$ be two component vertices in a graph stage $G^{(t)}$.
  Then the expected number of updates to vertices by
  $\textsc{Informremaining}$ when melding $u$ and $v$, assuming that all
  the sketch values are generated independently, is at most:
  \begin{align*}
  \frac{2\left|N^{G^{(t)}}_{remaining}(u)\right|
    \cdot \left|N^{G^{(t)}}_{remaining}(v)\right|}
  {\left|N^{G^{(t)}}_{remaining}(u)\right|
    +
    \left|N^{G^{(t)}}_{remaining}(v)\right|
  },
  \end{align*}
\end{lemma}

\begin{proof}
  Let's define:
  \begin{align*}
    n_{common} & = \left|N^{G^{(t)}}_{remaining}(u) \cap N^{G^{(t)}}_{remaining}(v)\right|,\\
    n_{u} & = \left|N^{G^{(t)}}_{remaining}(u) \setminus N^{G^{(t)}}_{remaining}(u)\right|,\\
    n_{v} & = \left|N^{G^{(t)}}_{remaining}(u) \setminus N^{G^{(t)}}_{remaining}(v)\right|.
  \end{align*} 
  If the minimum sketch value is generated by a vertex from
  $N^{G^{(t)}}_{remaining}(u) \cap N^{G^{(t)}}_{remaining}(v)$, then no cost is incurred.
  If it is generated by a vertex from
  $N^{G^{(t)}}_{remaining}(u) \setminus N^{G^{(t)}}_{remaining}(v)$, we need to update the
  every vertex in $N^{G^{(t)}}_{remaining}(v)$
  (line~\ref{algline:informvaboutu}).
  This happens with probability
  \begin{align*}
    &\ \frac{n_{u}}{n_{common} + n_{u} + n_{v}}
    \\\le&\ 
    \frac{n_{u} + n_{common}}{2n_{common} + n_{u} + n_{v}}
    \\=&\ 
    \frac{
      \left|N^{G^{(t)}}_{remaining}(u)\right|
    }{
      \left|N^{G^{(t)}}_{remaining}(u)\right| + \left|N^{G^{(t)}}_{remaining}(v)\right|
    }.
  \end{align*}
  Therefore the expected number of updates is bounded by:
  \begin{align*}
    \frac{\left|N^{G^{(t)}}_{remaining}(u)\right|
      \cdot \left|N^{G^{(t)}}_{remaining}(v)\right|}
    {\left|N^{G^{(t)}}_{remaining}(u)\right|
     +
     \left|N^{G^{(t)}}_{remaining}(v)\right|
    },
  \end{align*}
  and we get the other term (for updating $u$'s neighborhood) similarly.
\end{proof}

This allows us to carry out an amortized analysis for the number of updates to remaining
vertices.
We will define the potential function of an intermediate state during elimination
in terms of the degrees of component vertices \emph{in the original graph $G^{(0)}$}, in
which adjacent component vertices are \emph{not} contracted.
\begin{align*}
  \Phi(G^{(t)})
  \defeq
  \sum_{u \in V_{component}(G^{(t)})}D^{G^{(0)}}(u) \log \left(D^{G^{(0)}}(u)\right) ,
\end{align*}
where $D^{G^{(0)}}(u)$ for a component vertex $u\in V_{component}(G^{(t)})$ is defined to be
\begin{align*}
  D^{G^{(0)}}(u)=\sum_{v\in V(G),\textsc{Comp}(v)=u} \deg_{G^{(0)}}(u).
\end{align*}
This function starts out at $0$, and can be at most $m\log{n}$.

\begin{lemma}
  \label{lem:potential-decrease}
  The total potential decrease caused by turning remaining vertices into component vertices is
  at most $O(m\log{n})$.
\end{lemma}

\begin{proof}
  When we turn a remaining vertex $v$ into a component vertex, we decrease the value of
  $D(u)$ for every $u$ in $N^{G}_{remaining}(v)$.
  This causes a total potential decrease of at most
  \begin{align*}
    \left|N^{G}_{remaining}(v)\right|\log{n}.
  \end{align*}
  Since $|N^{G}_{remaining}(v)|$ is at most its degree, and each vertex can
  only be turned into a component vertex once, the total decrease in the potential is at most
  $O(m\log{n})$.
\end{proof}

\begin{lemma}
  \label{lem:potential-increase}
  When melding two neighboring component vertices in a graph $G^{(t)}$ to create
  $G^{(t+1)}$, we
  have
  \begin{align*}
    \expec{x_u: u \in V}{\text{number of remaining vertices updated}}
    \leq
    2\left(\Phi\left(G^{(t+1)}\right)-\Phi\left(G^{(t)}\right)\right).
  \end{align*}
\end{lemma}

\begin{proof}
  When melding two component vertices $u$ and $v$ in $G$ to form $G'$, the change in
  potential $\Phi(G')-\Phi(G)$ is given by
  \begin{align*}
    (D(u)+D(v))\log(D(u)+D(v))
    -D(u)\log D(u)
    -D(v)\log D(v).
  \end{align*}
  On the other hand, by Lemma~\ref{lem:remaining-updates} the expected number of
  remaining vertices updated is
  \begin{align*}
  \frac{2\left|N^{G^{(t)}}_{remaining}(u)\right|
    \cdot \left|N^{G^{(t)}}_{remaining}(v)\right|}
  {\left|N^{G^{(t)}}_{remaining}(u)\right|
    +
    \left|N^{G^{(t)}}_{remaining}(v)\right|
  }
    \leq
    \frac{2D(u)D(v)}{D(u)+D(v)}.
  \end{align*}
  
  Now it suffices to show the following the algebraic identity:
  \begin{align*}
    2 x \log{x} + 2 y \log{y} + \frac{2xy}{x + y}
    \le
    2 \left( x + y \right) \log\left( x + y\right),
  \end{align*}
  and let $x=D(u)$ and $y=D(v)$.
  By symmetry, we can assume $x\le y$ without loss of generality.
  Then we get
  \begin{align*}
    \frac{xy}{x + y}
    & \leq \frac{xy}{y}\\
    & = y \cdot \frac{x}{y}\\
    & \leq y \cdot \log\left(1 + \frac{x}{y} \right),
  \end{align*}
  where the last inequality follows from $\log(1 + z) \ge z$ when
  $z \le 1$.
  Plugging this in then gives:
  \begin{align*}
    2 x \log{x} + 2 y \log{y} + \frac{2xy}{x + y}
    & \leq 2 x \log{x} + 2 y \left( \log{y}
    + \log\left(1 + \frac{x}{y} \right) \right)\\
    & = 2 x \log{x} + 2 y \log\left( x + y\right)\\
    & \leq 2 \left( x + y \right) \log\left( x + y \right). \qedhere
  \end{align*}
\end{proof}

\begin{lemma}
  \label{lem:remaining-updates-from-meld}
  Over any fixed sequence of calls to \textsc{Meld}, the expected number of
  updates to the \textsc{Fill} heaps in remaining vertices (lines~\ref{algline:informvaboutu} and~\ref{algline:informuaboutv}) is bounded by $O(m \log{n})$.
\end{lemma}

\begin{proof}
  By Lemma~\ref{lem:potential-increase}, the number of updates is within a
  constant of the potential increase.
  Since our potential function $\Phi$ is bounded between $0$ and $O(m\log n)$,
  and by Lemma~\ref{lem:potential-decrease} does not decrease by more than
  $O(m\log n)$, the total number of updates is also bounded by $O(m\log n)$.
\end{proof}

Combining the above lemmas gives our main theorem from
Section~\ref{sec:Sketching} on maintaining one copy of the sketch.

\begin{proof}(of Theorem~\ref{thm:DataStructureMain})
  Given any graph $G$ and a fixed sequence of vertices for pivoting, we use the
  \textsc{PivotVertex} routine to produce the sequence of graph states
  \begin{align*}
    G^{(0)}=G,G^{(1)},G^{(2)},\dots,G^{(n)}.
  \end{align*}
  Recall that the goal is to maintain $\xmin(N^{G^{(t)}}_{remaining}(v))$ for
  all $v\in V_{component}(G^{(t)})$ and $\xmin(N^{G^{(t)}}_{fill}(v))$ for all
  $v\in V_{remaining}(G^{(t)})$.
  This is achieved by maintaining the two min-heaps, \textit{remaining} and
  \textit{Fill}.

  When pivoting a remaining vertex $v$, \textsc{PivotVertex} first removes it from the
  \textit{remaining} heaps among $v$'s component vertex neighbors
  (line~\ref{algline:remove-from-component}), which are at most as many as the original
  degree of $v$.
  Therefore the total cost of this part of the algorithm is $O(m\log n)$.
  The rest of the running time cost is incurred by updates to the
  \textit{Fill} heaps in \textsc{Informremaining}.
  By Lemma~\ref{lem:RemainingUpdatesDueToDeletion} and
  Lemma~\ref{lem:remaining-updates-from-meld}, the number of updates is bounded by
  $O(m\log n)$.
  As each update is a $O(\log n)$ operation on a heap, the the total running
  time is $O(m\log^2n)$.
  The final step of a meld consists of merging the \textit{remaining}
  heaps, and the cost of this step can be similarly bounded by $O(m\log^2 n)$.
\end{proof}

\section*{Acknowledgements}

We thank Animesh Fatehpuria and John Gilbert for many helpful
discussions regarding various topics in this paper.
In particular, we would like to acknowledge Animesh Fatehpuria
for independently obtaining the construction of the
covering set systems from Lemma~\ref{lem:CoveringSetSystem}.

\bibliographystyle{alpha}
\bibliography{references}

\begin{appendix}
	
\section{Proofs for Selection-Based Estimators}
\label{sec:SketchingProofs}

This section contains the proofs to Lemma~\ref{lem:minValue},
which claims that the reciprocal of the $\lfloor k(1-1/e) \rfloor$-ranked
element in $\variable{minimizers}(w)$ is a good approximation for the degree
of a vertex $w$.
The proofs follow similarly as in \cite{Cohen97} -
with the difference being the distribution of the
underlying $x$ variables.

We start by stating Hoeffding's tail bounds.

\begin{lemma}[Hoeffding's inequality]
\label{lem:Hoeffding}
  Let $b_1, b_2, \ldots, b_n$ be i.i.d. Bernoulli random variables such that
  $Pr[b_i = 1] = p$ and $Pr[b_i = 0] = 1-p$. Then,
  \begin{align*}
  Pr \left[ \sum_{1 \leq i \leq n} b_i \leq (p-\delta)n \right] \leq \exp(-2\delta^2 n), \\
  Pr \left[ \sum_{1 \leq i \leq n} b_i \geq (p+\delta)n \right] \leq \exp(-2\delta^2 n).
  \end{align*}
  
\end{lemma}

To apply Hoeffding's bounds, we will also need the following
numerical result.

\begin{lemma}\label{lem:expbounds}
  Let $-0.1 < \epsilon < 0.1$ and $d \geq 1$ be some parameter.
  Then we have:
  \[
  \exp\left( - 1 + \epsilon - \frac{2}{d} \right)
  \leq
  \left( 1 - \frac{1 - \epsilon}{d} \right)^{d}
  \leq \exp\left( - 1 + \epsilon \right)
  \]
\end{lemma}

\begin{proof}
The Maclaurin series for $\log(1 - x)$ is
\[
\log \left( 1 - x \right)
=
-x - \frac{x^2}{2} - \frac{x^3}{3} - \dots
\]
When $|x| \leq 0.1$, we have
\[
\left| \frac{x}{3} + \frac{x^2}{4} \ldots \right|
\leq
\frac{0.1}{3}
+ \frac{0.01}{4}
+ \ldots
\leq \frac{1}{2},
\]
so we have
\[
-x - x^2
\leq
\log \left( 1 - x \right)
\leq -x
\]

Apply this with
\[
x \leftarrow \frac{1 - \epsilon}{d}
\]
gives
\[
-\frac{1 - \epsilon}{d} - \frac{2}{d^2}
\leq
\log \left( 1 - \frac{1 - \epsilon}{d} \right)
\leq -\frac{1 - \epsilon}{d},
\]
which when exponentiated and taken to the $d\textsuperscript{th}$
power gives the result.
\end{proof}

We are now well equipped to prove Lemma~\ref{lem:minValue}
which, for the sake of convenience, we split into the
following two lemmas.

Let $k \geq \Omega(\log n \epsilon^{-2})$ denote
the number of copies of the $\ell_0$-sketch data structure.
Let $w$ be a vertex with degree $d > k$,
and let $q(w)$ denote the $\lfloor k(1-1/e) \rfloor$-ranked
element in $\variable{Minimizers}(w)$.

\begin{lemma}
\label{lem:minValueLB}
With high probability,
\[
\frac{1 - \epsilon}{d} \leq q(w).
\]
\end{lemma}

\begin{proof}
For any $1 \leq i \leq k$,
\[
\prob{}{ Minimizer(w)^{[i]} \geq \dfrac{1-\epsilon}{d} } 
= \prod_{v \in nbr(w)} Pr \left[ x_v^{[i]} \geq \dfrac{1-\epsilon}{d} \right]
= \left(1 -  \dfrac{1-\epsilon}{d} \right)^d.
\]

Let $J_i$ be the indicator variable which equals $1$ when $Minimizer(w)^{[i]} \geq \dfrac{1-\epsilon}{d}$
and equals $0$ otherwise.
So, $E[J_i] = \left(1 -  \dfrac{1-\epsilon}{d} \right)^d$
and,
\[
\prob{}{ q(w) \leq \frac{1 - \epsilon}{d} } = \prob{}{ \sum_{1 \leq i \leq k} I_i \leq k/e }.
\]

Using Lemma~\ref{lem:Hoeffding},
\[
\prob{}{ \sum_{1 \leq i \leq k} J_i \leq k/e } \leq \exp(-2k\delta^2) = \exp(-2 \log n (\delta / \epsilon)^{2})
\]
where $\delta = E[J_i] - 1/e$.
So,
\[
\delta/\epsilon = \dfrac{\left( 1 - \dfrac{1-\epsilon}{d} \right)^d - 1/e}{\epsilon}
\geq \dfrac{1/e^{(1-\epsilon+2/d)} - 1/e}{\epsilon}
\geq \dfrac{e^{\epsilon-2/d} - 1}{\epsilon \cdot e},
\]
where the first inequality follows from Lemma~\ref{lem:expbounds}.
Since $d > k \geq \Omega(\log{n} \epsilon^{-2})$, we have
$d \geq 4/\epsilon$, and can substitute to get:
\[
\delta/\epsilon \geq \dfrac{e^{\epsilon/2} - 1}{\epsilon \cdot e}
\geq \dfrac{1}{2e}.
\]
This gives us that
\[
\prob{}{ \sum_{1 \leq i \leq k} J_i \leq k/e } \leq (1/n)^c
\]
for some constant $c>0$.
\end{proof}

\begin{lemma}
	\label{lem:minValueUB}
	With high probability,
	\[
	\frac{1 + \epsilon}{d} \geq q(w).
	\]
\end{lemma}

\begin{proof}
	For any $1 \leq i \leq k$,
	\[
		\prob{}{ \textit{Minimizer}(w)^{[i]} \geq \dfrac{1+\epsilon}{d}}
		= \prod_{v \in nbr(w)} \prob{}{ x_v^{[i]} \geq \dfrac{1+\epsilon}{d} }
		= \left(1 -  \dfrac{1+\epsilon}{d} \right)^d.
	\]
	
	Let $I_i$ be the indicator variable which equals $1$ when $\textit{Minimizer}(w)^{[i]} \geq \dfrac{1+\epsilon}{d}$
	and equals $0$ otherwise.
	So, $E[I_i] = \left(1 -  \dfrac{1+\epsilon}{d} \right)^d$
	and,
	\[
	\prob{}{ q(w) \geq \frac{1 + \epsilon}{d}  } = Pr \left[ \sum_{1 \leq i \leq k} I_i \geq k/e \right]
	\]
	
	Using Lemma~\ref{lem:Hoeffding},
	\[
	\prob{}{ \sum_{1 \leq i \leq k} I_i \geq k/e } \leq \exp(-2k\delta^2) = \exp(-2 \log n (\delta / \epsilon)^{2}),
	\]
	where $\delta = 1/e - E[I_i]$.
	So,
	\begin{align*}
		\delta/\epsilon & = \dfrac{1/e - \left( 1 - \dfrac{1+\epsilon}{d} \right)^d}{\epsilon}
		\geq \dfrac{1/e - 1/e^{1+\epsilon}}{\epsilon} 
		\geq \dfrac{e^{\epsilon} - 1}{\epsilon \cdot e^{1+\epsilon}}
		\geq \dfrac{1}{e^{1+\epsilon}},
	\end{align*}
	where the first inequality follows from Lemma~\ref{lem:expbounds}.
	So, 
	$$Pr \left[ \sum_{1 \leq i \leq k} I_i \geq k/e \right] \leq (1/n)^c$$
	for some constant $c>0$.
\end{proof}

\end{appendix}

\end{document}